\documentclass[a4paper,thm-restate,11pt]{article}

\usepackage[utf8x]{inputenc}
\usepackage{latexsym}
\usepackage{amsmath}
\usepackage{verbatim}
\usepackage{enumerate}
\usepackage{amssymb}

\usepackage{csquotes}
\usepackage{graphicx}
\usepackage{textcomp}
\usepackage[ruled,vlined,linesnumbered]{algorithm2e}
\usepackage{epsfig}
\usepackage{amsfonts}
\usepackage{todonotes}
\usepackage[inline]{enumitem}
\usepackage{longtable}
\usepackage{xspace}
\usepackage{url}
\usepackage{subcaption}

\usepackage[colorlinks = true, backref=page]{hyperref}
\usepackage{amsthm}
\usepackage{amsmath}
\usepackage{amssymb}
\usepackage{amsfonts}
\usepackage{amstext}

\usepackage{xspace}
\usepackage{graphicx}
\usepackage[mathscr]{euscript}
\usepackage{fullpage}	

 \usepackage{thm-restate}

\usepackage{cleveref}
 \usepackage{authblk}

\theoremstyle{plain}
\newtheorem{theorem}{Theorem}
\newtheorem{lemma}{Lemma}
\newtheorem{observation}{Observation}

\theoremstyle{definition}
\newtheorem{definition}{Definition}

\newcommand{\CC}{{\mathcal C}}

\newcommand{\FF}{\mathcal{F}}

\newcommand{\disjPartitionedCompression
}[0]{{\sc r-${\mathcal H}$-SCC PC}}
\renewcommand{\S}{\mathcal{S}}

\renewcommand{\L}{\mathcal{L}}

\newcommand{\cH}{\mathcal{H}}
\newcommand{\HH}{\mathcal{H}}

\newcommand{\cW}{\mathcal{W}}
\newcommand{\cX}{\mathcal{X}}

\newcommand{\forwardshadow}{{\sf F\textrm{-}Shadow}}
\newcommand{\reverseshadow}{{\sf R\textrm{-}Shadow}}

\newenvironment{tightcenter}
 {\parskip=0pt\par\nopagebreak\centering}
 {\par\noindent\ignorespacesafterend}

\usepackage{tikz}
\usetikzlibrary{calc}
\usepackage{xargs}
\usepackage{xifthen}
\usepackage{framed}

\usepackage{ctable}

\newlength{\RoundedBoxWidth}
\newsavebox{\GrayRoundedBox}
\newenvironment{GrayBox}[1]%
   {\setlength{\RoundedBoxWidth}{\textwidth-4.5ex}
    \def\boxheading{#1}
    \begin{lrbox}{\GrayRoundedBox}
       \begin{minipage}{\RoundedBoxWidth}%
   }{%
       \end{minipage}
    \end{lrbox}%
    \begin{tightcenter}%
    \begin{tikzpicture}%
       \node(Text)[draw=black!90,fill=white,rounded corners,%
             inner sep=2ex,text width=\RoundedBoxWidth]%
             {\usebox{\GrayRoundedBox}};
        \coordinate(x) at (current bounding box.north west);
        \node [draw=white,rectangle,inner sep=3pt,anchor=north west,fill=white] 
        at ($(x)+(6pt,.75em)$) {\boxheading};
    \end{tikzpicture}
    \end{tightcenter}\vspace{0pt}%
    \ignorespacesafterend
}    

\newenvironment{problem}[2][]{\noindent\ignorespaces%
                                \FrameSep=6pt%
                                \parindent=0pt%
                \vspace*{-.5em}
                \ifthenelse{\isempty{#1}}{%
                  \begin{GrayBox}{\textsc{#2}}%                
                }{%
                }
                \newcommand\Prob{{\sf Problem:}}%
                \newcommand\Input{{\sf Input:}}%
                \newcommand\Parameter{{\sf Parameter:} }          
                \begin{tabular*}{\textwidth}{@{\hspace{.1em}} >{\itshape} p{1.2cm} p{0.85\textwidth} @{}}%     
            }{
                \end{tabular*}%
                \end{GrayBox}%
                \vspace*{-.5em}
                \ignorespacesafterend
            } 

\title{On the Parameterized Complexity of Deletion to $\mathcal{H}$-free Strong Components}

\author[1]{Rian Neogi}
\author[2]{M. S. Ramanujan}
\author[1,3]{Saket Saurabh}
\author[1]{Roohani Sharma}

\affil[1]{Institute of Mathematical Sciences, HBNI, India, \hspace{100pt} \{rianneogi,saket,roohani\}@imsc.res.in}
\affil[2]{University of Warwick, UK,   \hspace{100pt} r.maadapuzhi-sridharan@warwick.ac.uk}
\affil[3]{University of Bergen, Norway}
\date{}
\begin{document}
 \maketitle

\thispagestyle{empty}
\begin{abstract}
{\sc Directed Feedback Vertex Set (DFVS)} is a fundamental computational problem that has received extensive attention in parameterized complexity. 
In this paper, we initiate the study of a wide generalization, the {\sc $\cH$-free SCC Deletion} problem. Here, one is given a digraph $D$, an integer $k$ and the objective is to decide whether there is a vertex set of size at most $k$ whose deletion leaves a digraph where every strong component excludes graphs in the fixed finite family $\cH$ as (not necessarily induced) subgraphs. When $\cH$ comprises only the digraph with a single arc, then this problem is precisely DFVS.

Our main result is a proof that this problem is fixed-parameter tractable parameterized by the size of the deletion set if $\cH$ only contains rooted graphs or if $\cH$ contains at least one directed path. Along with generalizing the fixed-parameter tractability result for DFVS, our result also generalizes the recent results of G\"{o}ke et al. [CIAC 2019] for the {\sc 1-Out-Regular Vertex Deletion} and {\sc Bounded Size Strong Component Vertex Deletion} problems. Moreover, we design algorithms for the two above mentioned problems, whose running times are better and match with the best bounds for {\sc DFVS}, without using the heavy machinery of shadow removal as is done by G\"{o}ke et al. [CIAC 2019].
%give faster FPT algorithms for the two above mentioned problems, thereby improving the results of G\"{o}ke et al. [CIAC 2019].
\end{abstract}

%\tableofcontents

\newpage
\setcounter{page}{1}
\section{Introduction}\label{sec:intro}
In the {\sc Directed Feedback Vertex Set (DFVS)} problem, the input is a digraph $D$ and an integer $k$ and the objective is to decide whether there is a set $X\subseteq V(D)$ of size at most $k$ such that $D-X$ is acyclic. DFVS is a fundamental computational problem that has received extensive attention in various subdomains of algorithmics. The parameterized complexity of this problem was a long standing open problem in the area until Chen et al.~\cite{ChenLLOR08} gave a fixed-parameter tractable (FPT) algorithm with running time $O(k!4^kk^4nm)$. Here, $n$ and $m$ denote the number of vertices and arcs in the digraph respectively. Although subsequent work~\cite{LokshtanovR018} has improved the dependence on the input size to linear, 
it remains an open problem whether the $2^{O(k\log k)}$ dependence on $k$ is asymptotically the best possible.

The result of Chen et al.~and the techniques used therein also helped kick off a line of research in parameterized complexity where the goal is to understand how far the fixed-parameter tractability of DFVS can be extended to various generalizations of DFVS. Chitnis et al.~\cite{ChitnisCHM15} obtained an FPT algorithm for the {\sc Subset DFVS} problem, where the goal is to delete at most $k$ vertices that intersect all directed cycles passing through a specified subset of vertices. A general and abstract formulation of the powerful directed {\em shadow removal} technique first designed by Chitnis et al.~\cite{ChitnisHM13}, was developed in this work and it has found several applications in subsequent work~\cite{KratschPPW15,ChandrasekaranM17,LokshtanovRSZ17,GokeMM19}. %We refer the reader to \cite{ChitnisH16} for a survey on this technique.
Lokshtanov et al.~\cite{LokshtanovRSZ17} studied the {\sc Directed Odd Cycle Transversal} problem where the objective is to delete at most $k$ vertices that intersect all directed {\em odd} cycles in the given digraph. They proved that this problem is W[1]-hard and so is unlikely to admit an FPT algorithm. Moreover, they used the shadow removal technique to obtain a fixed-parameter 2-approximation algorithm for this problem.
More recently, G\"{o}ke et al.~\cite{GokeMM19} studied the problems of deleting at most $k$ vertices to (i) obtain a digraph where every strong component is of bounded size and (ii) obtain a digraph where every strong component induces a graph where every vertex has out-degree exactly 1, i.e. is a 1-out-regular digraph.

	In this paper, we extend this line of research by initiating the study of a wide generalization of the problems studied by G\"{o}ke et al., which we call the {\sc $\cH$-Strong Connected Component Deletion} ($\cH$-SCC Deletion) problem and define below. Here, $\cH$ is a fixed finite family of digraphs.

\vspace{-5 pt}

 \begin{problem}[]{$\cH$-SCC Deletion}
  \Input & A digraph $D$, an integer $k$. \\
  \Parameter & \hspace{-2pt} $k$ \\
  \Prob  & Does there exist a set $S$ of at most $k$ vertices such that no strong component of $D$-$S$ contains a graph in $\cH$ as a subgraph?
\end{problem}

\vspace{-5 pt}

In all our results, $n$ denote the number of vertices in the input graph and $h= \max_{H \in \cH} \allowbreak  \vert V(H) \vert$.
%Graphs in $\cal H$ may also be disconnected. 
%A digraph $H$ is said to be {\em rooted} at $v\in V(H)$ if $H$ contains an arborescence rooted at $v$. That is, there is a directed $v$-$w$ path in $H$ for every $w\in V(H)$. A digraph $H$ is simply said to be {\em rooted} if it is rooted at some vertex.  
{\sc Rooted $\cH$-SCC Deletion} ({\sc r-$\cH$-SCC Deletion}) denotes the special case of {$\cH$-SCC Deletion} where every graph in $\cH$ contains an arborescence. An arborescence is a rooted directed tree where every vertex except the root has in-degree exactly $1$ and the root has in-degree $0$. Notice that {\sc r-$\cH$-SCC Deletion} already generalizes several problems described above including DFVS ($\cH$  comprises the graph with a single arc), obtaining  strong components of size at most $s$ ($\cH$ comprises all arborescences of size $s+1$) and obtaining strong components with out-degree at most 1 ($\cH$ comprises the star with two leaves and both arcs oriented away from the centre). Our main result gives a unified proof of the fixed-parameter tractability of these problems. 

\begin{restatable}{theorem}{rootedSCCDeletion}\label{thm:rootedSCCDeletion}
{\sc  r-$\cH$-SCC Deletion}  can be solved in time $2^{O(k^3 \log k)} \cdot n^{O(h)}$.
%, where $h= \max_{H \in \cH} \vert V(H) \vert$.
\end{restatable}

\vspace{-4 pt}

Theorem~\ref{thm:rootedSCCDeletion} also holds for $\cH$-SCC Deletion in the case where, for every graph in $\cH$ there is a vertex that is reachable from every other vertex. One can infer this by simply reversing both the input graph and the forbidden graphs and applying the main theorem.
%The $O^*$ notation hides polynomial factors in $n$, the number of vertices of the input graph.
We also remark that in general, the $n^{O(h)}$ dependence in the running time of the algorithm of Theorem~\ref{thm:rootedSCCDeletion} is very likely unavoidable.
% That is, for {\sc r-$\cH$-SCC Deletion}, a fixed-parameter algorithm parameterized by $k$ and some $g({\cH})$ is unlikely to exist. 
 Indeed, consider the following reduction from the {\sc Clique} problem where the input is an undirected graph~$G$ and $\ell\in {\mathbb N}$, and the objective is to decide whether~$G$ contains a clique of size~$\ell$. We orient all edges in~$G$ arbitrarily, add a universal sink vertex~$v^\star$ and then a universal source vertex~$u^\star$ and the arc $(v^\star,u^\star)$ to obtain a strongly connected digraph, set $k=0$, and set $\cH$ to be the set of all tournaments on exactly $\ell+2$ vertices. Then, an FPT algorithm for {\sc r-$\cH$-SCC Deletion} parameterized by $k$ and $\ell$ would imply an FPT algorithm for {\sc Clique} parameterized by $\ell$.
%  which cannot exist unless FPT=W[1]. 

When $\cH$ only comprises of the star with $d+1$ leaves with all arcs oriented away from the centre, a closer inspection of the algorithm of Theorem~\ref{thm:rootedSCCDeletion} demonstrates that
% the algorithm of Theorem~\ref{thm:rootedSCCDeletion} 
it can be implemented in a way that implies a fixed-parameter algorithm parameterized by both~$k$ and~$d$ for this problem. We call this problem, $d$-{\sc Out-Degree SCC Deletion}. In this problem, the objective is to decide whether $k$ vertices can be deleted from a given digraph to ensure that the graph induced by each strong component has out-degree at most $d$.
%  is FPT parameterized by $k$, for {\em fixed} $d$. However,  
%To achieve this improvement, we essentially rely on the fact that  a vertex of out-degree $d+1$ can be computed in polynomial time.

 \begin{theorem}\label{thm:outDegreeSCCDeletion}
 	$d$-{\sc Out-Degree SCC Deletion} can be solved in time $2^{O(k^3 \log k)} \cdot n^{O(1)}$.
 	%, where the exponent of $n$ is independent of $k$ and $d$.
 \end{theorem}
\vspace{-4 pt}

Our next result concerns the 
 {\sc Path $\cH$-SCC Deletion} ({\sc p-$\cH$-SCC Deletion}) problem, which is the special case of {$\cH$-SCC Deletion} where $\cH$ contains at least one directed path. We show that with an appropriate fixed-parameter preprocessing routine, this problem can be reduced to {\sc r-$\cH$-SCC Deletion} where $\cH$ only comprises of the path of length $g(\cH)$ for some function $g$. Invoking Theorem~\ref{thm:rootedSCCDeletion} then leads us to the following result.
% , the length of the longest path in $\cH$.  

\begin{restatable}{theorem}{pathSCCDeletion}\label{thm:pathSCCDeletion}
{\sc p-$\cH$-SCC Deletion} can be solved in time $2^{O(k^3 \log k)} \cdot h^{O(k)} \cdot 2^{O(h^6)} \cdot n^{O(h^3)}$.
%, where $h= \max_{H \in \cH} \vert V(H) \vert$.
\end{restatable}

\vspace{-4 pt}

%In this section, we give an algorithm for 
We then pay special attention to the {\sc r-$\cH$-SCC Deletion}  
%Theorem~\ref{thm:rootedSCCDeletion} 
problem 
%that runs in time $O^*(2^{O(k \log k)})$
 when $\cH$ contains only the out-directed {2-star}, i.e., the $1$-{\sc Out-Degree SCC Deletion} problem. Notice that a strongly connected graph with at least two vertices that excludes this graph as a subgraph must be a simple cycle, and so is 1-out-regular. A 1-out-regular digraph is a digraph where every vertex has out-degree exactly $1$. Therefore, this special case of {\sc  r-$\cH$-SCC Deletion} is precisely the {\sc 1-Out-regular Deletion} problem where one is given a digraph $D$ and an integer~$k$ and the objective is to decide whether there is a set of vertices of size at most~$k$ whose deletion leaves a digraph where every strong component induces a 1-out-regular subgraph. G{\"{o}}ke et al.~\cite{GokeMM19} recently gave an algorithm for this problem with running  time $2^{O(k^3)} \cdot n^{O(1)}$. We give an improved algorithm for this problem with an 
% Moreover, our 
 asymptotic dependence on $k$ that matches that of the current best algorithm for DFVS~\cite{ChenLLOR08}, which is a special case of {\sc 1-Out-regular Deletion}.

\begin{restatable}{theorem}{oneOutRegularDeletion}\label{thm:oneOutRegularDeletion}
	{\sc 1-Out-regular Vertex Deletion} can be solved in time $2^{O(k\log k)} \cdot n^{O(1)}$.
\end{restatable}

\vspace{-4 pt}

%In this section, we give an algorithm for {\sc Rooted $\cH$-SCC Deletion} that runs in time $O^*(2^{O(k\log s})$ in 
%
Finally, we also study
the special case of {\sc r-$\cH$-SCC Deletion} when $\cH$ is the set of all arborescences on exactly $s+1$ vertices.
Notice that the strongly connected graphs that exclude the graphs in $\cH$ as subgraphs are precisely the strongly connected graphs of size at most $s$. This problem when $\cH$ is the set of all arborescences on exactly $s+1$ vertices is called the {\sc Bounded Size Strong Component Vertex Deletion (BSSCVD)}. 
We improve upon the result of G{\"{o}}ke et al.~\cite{GokeMM19} who gave an algorithm for {\sc BSSCVD} with running time $4^k(ks+k+s)!\cdot n^{O(1)}$.
% Notice that when $s=1$, this is precisely the DFVS problem.

\begin{restatable}{theorem}{boundedSizeDeletion}\label{thm:boundedSizeDeletion}
	{\sc BSSCVD} can be solved in time $2^{O(k(\log k+\log s))} \cdot n^{O(1)}$.
\end{restatable}

We now give an overview of the techniques used to prove our results. 
%Missing proofs can be found in the appended full version.

%Our main goal is to study the parameterized complexity of these problems and so, we do not attempt to optimize the polynomial dependence on the input size in our algorithm design and analysis. 

% 
%
%\paragraph*{Our Techniques.} 

\noindent
\textbf{Algorithm for {\sc r-$\cH$-SCC Deletion}}.~ 
We begin by using the technique of iterative compression to obtain a  tuple $(D,{\cal S}=(S_1,\dots, S_q),W,k)$ such that $W$ is a solution for the instance $(D,k+1)$ of {\sc r-$\cH$-SCC Deletion}, $S_1,\dots, S_q$ partition $W$ and moreover,  if there is a solution for $(D,k)$, then there is a solution that is disjoint from $W$ and  intersects $S_i$-$S_j$ paths for $i<j$. We note that this step is standard when dealing with directed cut problems.

%
%\vspace{-10 pt}
A commonly used technique subsequent to this step (albeit one that we {\em do not} employ) is the directed {\em shadow removal} technique introduced by Chitnis et al.~\cite{ChitnisHM13} where one identifies a set of vertices $Z$ such that for some hypothetical solution $X$, $Z$ is disjoint from $X$ and contains the set of vertices that are either unable to reach $W$ or are unreachable from $W$ in $D-X$. This set is then removed in a problem specific way while preserving all obstructions.  While this can be easily achieved for certain simple obstructions, we are dealing with an arbitrary family of digraphs with the only assumption being that they are rooted. Consequently, it is not at all clear how one could implement the removal of vertices in $Z$ and that makes our task significantly more challenging. To avoid this obstacle, we forgo the technique of shadow removal and directly design an intricate branching algorithm.

 The crux of this algorithm is the observation that for a special type of solution $X$, for every forbidden subgraph $F$, either $X$ intersects $V(F)$ or $X$ contains an $S_1$-$\{r(F)\}$ separator ($r(F)$ denotes a fixed root of $F$) or $X$ contains a $\{u\}$-$W$ separator for some vertex $u \in V(F)$. We then show that there is always an efficiently computable forbidden subgraph upon which we can branch exhaustively using the above observation in such a way that we always make progress. The fact that such a subgraph can always be identified efficiently is far from obvious and proving it is one of our main technical challenges. 

%\vspace{-4 pt}

\noindent
\textbf{Algorithm for {\sc  p-$\HH$-SCC Deletion}.}~We show that for every finite family of digraphs~$\HH$, there exists another  (infinite) family $\HH^*$, such that $\HH$-SCC problem is equivalent to the problem of deleting at most $k$ vertices to exclude all graphs in $\HH^*$ as subgraphs in the remaining graph (not necessarily in a single strong component). 
Moreover, we show that when $\HH$ contains a directed path, then the family $\HH^*$ can be partitioned into two, say $\HH_1^*$ and $\HH_2^*$ such that $\HH_1^*$ is finite and the problem of   deleting at most $k$ vertices to exclude all graphs in $\HH_2^*$ as subgraphs in the remaining graph  is equivalent to the {\sc r-$\cX$-SCC Deletion} where $\cX$ only contains a directed path whose length depends on $\cH$. This allows us to branch on all subgraphs isomorphic to graphs in $\HH_1^*$ and then invoke Theorem~\ref{thm:rootedSCCDeletion}.
%as described above, exhibits certain properties that can be exploited to yield the FPT algorithm for the {\sc Path $\HH$-SCC} problem.
%\vspace{-6 pt}

\noindent
\textbf{Improved algorithms for {\sc  1-Out-Regular Vertex Deletion} and {\sc  BSSCV}.}~Here, we begin in the same way as for {\sc r-$\HH$-SCC Deletion} by obtaining a  tuple $(D,{\cal S}=(S_1,\dots, S_q),W,k)$ such that if there is a solution for $(D,k)$, then there is a solution that is disjoint from $W$ and  intersects all $S_i$-$S_j$ paths for $i<j$.
In the case of {\sc  1-Out-Regular Vertex Deletion}, the main new contribution  that results in a speedup over Theorem~\ref{thm:rootedSCCDeletion} is a lemma that shows that if we consider an $S_1$-$W\setminus S_1$ separator $C$ such that every vertex reachable from $S_1$ in $D-C$ has out-degree at most $1$ in~$D$, then there is a solution whose intersection with this set of reachable vertices is contained within an efficiently computable set of $O(k)$ vertices. This gives us a branching algorithm where we essentially compute a ``furthest'' $S_1$-$W\setminus S_1$ separator $C$ of size at most $k$ in time $2^{O(k)} \cdot n^{O(1)}$ and then branch on deleting a vertex of $C$ or one of these $O(k)$ vertices in the reachable set. In the case of {\sc BSSCV}, we show that if for some $x\in S_1$, $D$ contains a subgraph of size $s+1$ with an arborescence rooted at $x$, then at least one of these vertices must either be deleted or must have all its paths to $W$ deleted when removing a solution. In the latter case, we will be able to branch on an $\{x\} $-$W$ important separator~\cite{marx2006parameterized} of size at most $k$.

\noindent
\textbf{Further Remarks.}~
%Our algorithms improve the results of G\"{o}ke et al.~\cite{GokeMM19}. 
The results of G\"{o}ke et al.~\cite{GokeMM19} crucially use the technique of {\em covering the shadow} which adds a factor of $2^{O(k^2)} \cdot n^{O(1)}$ to the running time of any algorithm that uses it. Thus, our $2^{O(k \log k)} \cdot n^{O(1)}$ algorithm (Theorem~\ref{defn:extdefn}) is an improvement over what is currently possible using the shadow covering technique. 
Moreover, although all our results are stated for the vertex deletion version of the problems, we would like to mention that these results will apply for the corresponding arc deletion versions of the problems as well, as all our proofs go through for the later case also. 

\section{Preliminaries}
We use standard notion regarding digraphs. Given two graphs $D_1,D_2$ we denote the digraph $D_1 \cup D_2$ as the digraph with vertex set $V(D_1) \cup V(D_2)$ and arc set $A(D_1) \cup A(D_2)$. Note that the vertex sets of $D_1$ and $D_2$ are not necessarily disjoint. By $D_1 \subseteq D_2$, we mean that $D_1$ is a {\em subdigraph} of $D_2$. 
 We use~$|D|$ as a shorthand for $|V(D)|$. A {\em strongly connected component (or strong component)} of a digraph $D$ is a maximal set $S\subseteq V(D)$ such that for any pair $u,v$ of vertices in $S$, there is a path from $u$ to $v$ and from $v$ to $u$ in $D$. It is well known that it is possible to order the strongly connected components of $D$ such that there is no path from a component to another component behind it in the ordering. This is called the {\em topological ordering} of the strongly connected components of $D$. For any digraph $D$, for any $X \subseteq V(D)$, $D[X]$ denotes the subdigraph of $D$ {\em induced} by $X$ and $D-X$ denotes the induced subdigraph $D[V(D)\setminus S]$.
For a vertex $v \in V(D)$, by $N_D^+(v)$ we denote $\{u \in V(D)  \mid (v,u) \in A(D)\}$ and by $N_D^-(v)$ we denote $\{u \in V(D)  \mid (u,v) \in A(D)\}$. 
For a subset $S \subseteq V(G)$, by $N_D^+(S)$ we denote the set $\bigcup_{v \in S} N_D^+(v) \setminus S$. By $N_D^+[S]$ we denote $N_D^+(S) \cup S$. Analogous definitions hold for $N_D^-(S)$ and $N_D^-[S]$. When the digraph $D$ is clear from the context, we drop the subscript $D$ from the notation.

 For any $S,T \subseteq V(D)$, by an {\em $S$-$T$ path} in $D$ we mean a path from some vertex of $S$ to some vertex of $T$ in $D$. 
 When $S=\{v\}$ (resp.~$T=\{v\}$) is a singleton set, we write a $v$-$T$ (resp.~$S$-$v$ path).   
For $S,T,C \subseteq V(D)$ such that $S \cap T = \emptyset$, we say that $C$ is an {\em $S$-$T$ separator} if there is no directed $S$-$T$ path in $D- C$ and $C \cap S = C \cap T = \emptyset$. For an $S$-$T$ separator $C$, by $R_D(S,C)$ we denote the set of vertices $v$ such that there exists an $S$-$\{v\}$ path in $D-C$. By $\overline{R}_D(S,C)$ we denote the set $V(D)\setminus R_D(S,C)$, that is the set of vertices that are unreachable from $S$ after removing $C$. Note that for any set $R \subseteq V(D)$ such that $S \subseteq R$, $R \cap T = \emptyset$ and $N^+(R) \cap T =\emptyset$, {if $R$ is reachable from $S$} in $D[R]$, then the set $C=N^+(R)$ is an $S$-$T$ separator such that $R_D(S,C) = R$. We say that an $S$-$T$ separator $C$ {\em covers} an $S$-$T$ separator  $C'$ if $R_D(S,C) \supseteq R_D(S,C')$. We say that $C$ {\em tightly covers} $C'$ if $C$ covers $C'$ and there does not exist a $C''$ that covers $C'$ and is covered by $C$. Two $S$-$T$ separators are {\em incomparable} if neither covers the other. 
$\lambda_D(S,T)$ denotes the size of a minimum $S$-$T$ separator  in $D$. 
It is well known~\cite{ChenLLOR08} 
%that as a consequence of the above observation, 
that there exists a unique minimum $S$-$T$ separator $C_{\rm min}(S,T)$ and a unique minimum $S$-$T$ separator $C_{\rm max}(S,T)$ such that for every minimum $S$-$T$ separator $C$, $C_{\rm min}(S,T)$ is covered by $C$ and $C_{\rm max}(S,T)$ covers $C$. We call $C_{\rm min}(S,T)$ the {\em closest minimum $S$-$T$ separator}  and $C_{\rm max}(S,T)$ the {\em furthest minimum $S$-$T$ separator}. 
Moreover, we define $R_{\rm min}(S,T)=R_D(S,C_{\rm min}(S,T))$ and $R_{\rm max}(S,T)=R_D(S,C_{\rm max}(S,T))$. All four of these sets can be computed in polynomial time using  max-flow computations (see \cite{marx2006parameterized}).  

An {\em arborescence} is a rooted directed tree where every vertex except the root has in-degree exactly $1$ and the root has in-degree $0$. 
A digraph $D$ is said to be {\em rooted} at $v\in V(D)$ if $D$ contains as a subdigraph on $V(D)$ an arborescence rooted at~$v$. That is, there is a directed $v$-$w$ path in $D$ for every $w\in V(D)$. A digraph $D$ is simply said to be {\em rooted} if it is rooted at some vertex. By $r(D)$, we denote the vertex that is the root of $D$. If there are multiple roots for~$D$, we canonically fix one vertex for $r(D)$. 
For a digraph $D$ and a family of digraphs $\cal H$ (potentially containing rooted digraphs), we say that a subset $S \subseteq V(D)$ is {\em $\cal H$-free} if $D[S]$ does not contain any graph in $\cal H$ as a subgraph. When $S=V(D)$, we say that $D$ is $\cH$-free.
In the case when every graph in $\cal H$ is a rooted graph, we say that $S$ is {\em root-$\cal H$-free} if the root of every subgraph of $D$ that is isomorphic to a graph in $\cal H$ is not contained in $S$. Observe that if a set $S$ is root-$\cal H$-free then it is also $\cal H$-free.
We say that a set $X \subseteq V(D)$ is an {\em $\cal H$-deletion set} for $D$ if there is no subgraph isomorphic to a graph in $\cal H$ that is contained in any strong component of $D-X$. Furthermore, we say that~$X$ is a {\em solution} for the tuple $(D,k)$ if $X$ is an $\cal H$-deletion set and $|X|\le k$. Throughout the paper, $h=\max_{H \in \cH} \vert V(H) \vert$.

The following observation follows from the sub-modularity of separators.
\begin{observation}[\rm \cite{cygan2015parameterized}]\label{obs:submod}
	Let $C_1$ and $C_2$ be two minimum $S$-$T$ separators in a digraph $D$. Let $R_1 = R_D(S,C_1)$ and $R_2 = R_D(S,C_2)$, then $N^+(R_1 \cap R_2)$ and $N^+(R_1 \cup R_2)$ are also minimum $S$-$T$ separators of $D$.
\end{observation}

\begin{lemma}
\label{lem:tightSuccessorComputation}{\rm \cite{LokshtanovR018}}
	There is a polynomial-time algorithm that, given a digraph $D$ and an $S$-$T$ separator  $C$ in $D$, either correctly concludes that $C=R_{\rm max}(S,T)$ or outputs a minimum $S$-$T$ separator  that tightly covers $C$. 
\end{lemma}

\begin{lemma}\label{lem:tightness}{\rm \cite{LokshtanovR018}}
	Let $C_1=N^+(R_1)$ and $C_2=N^+(R_2)$ be minimum $S$-$T$ separators such that $C_2$ tightly covers $C_1$. Then there does not exist any minimum $S$-$T$ separator $C$ such that $C \cap (R_2\setminus N[R_1]) \neq \emptyset$.
\end{lemma}

Lemma~\ref{lem:tightness} basically implies that there is no minimum $S$-$T$ separator in $D$ that contains a vertex in $R_D(S,C_2) \setminus N[R_D(S,C_1)]$. 
%The following two lemmas follow from the above.

\begin{lemma}\label{lem:closest_sep_lambda_inc}{\rm\cite{kratsch2012representative}}
	Let $D$ be a digraph and let $S,T\subseteq V(D)$ such that $S \cap T = \emptyset$. Let $C$ be the closest (resp. furthest) $S$-$T$ separator in $D$ and let $v$ be a vertex in $R_D(S,C)$ (resp. $\overline{R}_D(S,C)$). Then every $S$-$(T \cup \{v\})$ (resp. $(S \cup \{v\})$-$T$) separator is of size strictly greater than $\lambda _D(S,T)$. 
\end{lemma}

\begin{lemma}\label{lem:tight_sep_lambda_inc}
	Let $D$ be a digraph and let $S,T$ be disjoint subsets of $V(D)$. Let $C_1=N^+(R_1)$ and $C_2=N^+(R_2)$ be two minimum $S$-$T$ separators such that $C_2$ tightly covers $C_1$ and $C_1\neq C_2$. Let $u \in C_1$ and $v \in R_2\setminus N[R_1]$. Then every $(S \cup \{u\})-(T \cup \{v\})$ separator is of size strictly greater than $\lambda _D(S,T)$.  
\end{lemma}

\begin{proof}
	Let $C'$ be an $(S \cup \{u\})$-$(T \cup \{v\})$ separator. Since $u \in N[R_1]$ and $v \in R_2 \setminus N[R_1]$, it must be the case that $C'$ contains a vertex in $R_2\setminus N[R_1]$. Thus it follows from Lemma~\ref{lem:tightness} that $C'$ is of size greater than $\lambda _D(S,T)$.    	
\end{proof}

\begin{definition}{\rm\cite{ChenLLOR08,cygan2015parameterized,marx2006parameterized}}
	For a digraph $D$ and disjoint subsets $S,T$ of $V(D)$, an $S$-$T$ separator $C$ is said to be {\em important} if there is no $S$-$T$ separator $C' \neq C$ such that $|C'| \le |C|$ and $C'$ covers $C$.
\end{definition}

\begin{lemma}\label{lem:imp_sep_bound}{\rm\cite{ChenLLOR08,cygan2015parameterized,marx2006parameterized}}
	There are at most $4^k$ important separators of size at most $k$. Moreover, they can be enumerated in $O^*(4^k)$ time.
\end{lemma}

The notion of important separators has been extrememly useful in the design of parameterized algorithms (see Marx' survey~\cite{Marx2011ImportantSA}).

\begin{lemma}\label{lem:pushing_property}
Let $D$ be a digraph and $S,T \subseteq V(D)$ such that $S \cap T = \emptyset$. Let $X\subseteq V(D)$ contain a minimal $S$-$T$ separator $C$. Let $C'$ be an $S$-$T$ separator that covers $C$ and let $X' = (X \setminus C) \cup C'$. Let $u,v \in \overline{R}_D(S,X')$. If there exists a path from $u$ to $v$ in $D-X'$ then there exists a path from $u$ to $v$ in $D-X$.
\end{lemma}
\begin{proof}
	Suppose not. Then there exists a path $P$ from $u$ to $v$ in $D-X'$ but no path from $u$ to $v$ in $D-X$, that is, $V(P) \cap X' = \emptyset$ and $V(P) \cap X \neq \emptyset$. Let $a \in V(P) \cap X$. Then $a \in X \setminus X' \subseteq C$. Let $R=R_D(S,X)$. Since $C$ is a minimal $S$-$T$ separator and is contained in $X$, we have that $C=N^+(R)$ and thus $a \in N^+(R)$. Since $C'$ covers $C$, $R_D(S,X) \subseteq R_D(S,X')$. Also, since $a \not\in X'$, it follows that $a \in R_D(S,X')$. Thus, there is a path from $S$ to $a$ in $D-X'$, which together with the subpath of $P$ from $a$ to $v$ gives a walk (and eventually a path) from $S$ to $v$ in $D-X'$. This is a contradiction to the fact that $v \in \overline{R}_D(S,X')$.
\end{proof}

\section{The FPT algorithm for rooted {\sc $\cH$-SCC Deletion}}

In this section, we consider the $\HH$-SCC deletion problem when all the graphs in $\HH$ are rooted (the {\sc r-$\HH$-SCC} problem) and prove Theorem~\ref{thm:rootedSCCDeletion}. Towards the end of the section, we also exhibit the proof of Theorem~\ref{thm:outDegreeSCCDeletion}.

\rootedSCCDeletion*

We use the standard technique of Iterative Compression (\cite{reed2004finding}) to reduce the task of solving our instance of {\sc r-$\cH$-SCC Deletion} to that of solving at most $2^{k+1}n$ instances of the {\sc Disjoint r-$\cH$-SCC Deletion Compression (r-$\cH$-SCC DC)} problem, where we are given a digraph~$D$ and a solution~$W$ of size at most~$k+1$ and the goal is to compute a solution of size at most~$k$ that is disjoint from~$W$, if one exists.
 The rest of the section is devoted to the proof of the following lemma, which as described above proves Theorem~\ref{thm:rootedSCCDeletion}.

\begin{lemma}\label{lem:disjCompressionSolution}
There is an algorithm that, given an instance $(D,W,k)$ of 
	{\sc r-$\cH$-SCC DC}, runs in time $2^{O(k^{3}\log k)} \cdot n^{O(h)}$ and either correctly concludes that there is no solution for this instance disjoint from $W$ or outputs a solution disjoint from $W$.
\end{lemma}

\subsection{Reduction to Nice Instances of the Partitioned Compression Version}

We further solve {\sc r-$\cH$-SCC DC} by making $2^{O(k\log k)}$ calls to a subroutine that is allowed to assume the existence of a specific type of solution for instances of the {\sc r-${\mathcal H}$-SCC Partitioned Compression} 
 ({\disjPartitionedCompression
}) problem, which is described below.
An instance of {\disjPartitionedCompression} is a tuple $(D,\S=(S_1,\dots, S_q),W,k)$, where $W$ is a solution for the instance $(D,k+1)$ of $\HH$-SCC
and~$\S$ is an ordered partition of $W$. 
A set $X\subseteq V(D)$ is said to be a {\em solution} for the instance $(D,\S=(S_1, \ldots, S_q), W,k)$ of {\disjPartitionedCompression} if $X$ is a solution for the instance $(D,k)$ of ${\cal H}$-SCC, $X \cap W = \emptyset$, and $X$ intersects all $S_i$-$S_j$ paths in $D$, for every $j>i$. Observe that every solution of the instance $(D,\S,W,k)$ of {\disjPartitionedCompression}, is also a solution of $(D,W,k)$ of {\sc r-$\HH$-SCC DC}. 
We now define a special kind of solution for {\disjPartitionedCompression}, which we call a {\em nice solution}, and as we will see soon, it turns out that it is enough to look for nice solutions for our purpose.

\begin{definition}[Nice Instances]\label{def:niceInstances}
Let $(D,\S=(S_1,\dots, S_q),W,k)$ be an instance of {\disjPartitionedCompression}. A solution~$X$ for this instance is said to be {\em nice} if for every subgraph $F \subseteq D$ such that~$F$ is isomorphic to a graph in~$\cal H$, and each $i \in [q]$, one of the following holds: 

	\begin{enumerate}
	\item $X$ intersects $V(F)$,
	\item $r(F)\notin R(S_i,X)$, or
	\item $\exists v\in V(F)$ such that there is no $\{v\}$-$S_i$ path in $D-X$.
	\end{enumerate}
	
\end{definition}

Observe from the definition above, that if $X$ is a solution of {\disjPartitionedCompression} such that after its deletion each $S_i$ is in exactly one strong component, then $X$ is a nice solution. 

\begin{observation}\label{obs:nice_hereditary1}
Let $(D,\S=(S_1,\dots, S_q),W,k)$ be an instance of {\disjPartitionedCompression}. If $X$ is a nice solution for this instance, then for any $X' \subseteq X$, $X\setminus X'$ is a nice solution for the instance $(D-X',\S,W,k-\vert X' \vert)$.
\end{observation}

\begin{observation}\label{obs:nice_hereditary2}
Let $X' \subseteq V(D)\setminus W$. If $X$ is a nice solution for $(D-X',\S,W,k-\vert X' \vert)$ then $X \cup X'$ is a nice solution for $(D,\S,W,k)$.
\end{observation}

We now show that for our purposes, it is enough to look for nice solutions.

\begin{lemma}\label{lem:DisjCompressionToDisjPartionedCompression}
$(D,W,k)$ is a yes-instance of {\sc r-$\cH$-SCC DC} if and only if $\allowbreak (D, \S =(S_1, \ldots, S_q),\allowbreak W,k)$ is a nice yes-instance of {\disjPartitionedCompression}, for some ordered partition $\S=(S_1, \ldots, S_q)$ of $W$.
\end{lemma}
\begin{proof}
		For the forward direction, suppose $(D,W,k)$ is a yes-instance of {\sc r-$\cH$-SCC DC}. Let $X \subseteq V(D) \setminus W$ be an $\HH$-SCC solution for $D$ of size at most $k$. Let $C_1, \ldots, C_q$ be the set of strongly connected components of $D-X$ as they appear in some topological ordering of the strongly connected components of $D-X$. For each $i \in [q]$, let $S_i = W \cap C_i$. Since $W \subseteq V(D) \setminus X$, $(\S=(S_1, \ldots, S_q))$ is an ordered partition of $W$.
From the construction above, note that $X$ is a solution to the instance $(D,\S,W,k)$ of {\disjPartitionedCompression}. We now prove that $X$ is also nice.
Let $F \subseteq D$ such that $F$ is isomorphic to a graph in $\cH$. Fix $i \in [q]$. We will prove that either $V(F) \cap X \neq \emptyset$ or $X$ contains either an $S_i$-$\{r(F)\}$ separator or an $\{u\}$-$S_i$ separator for some vertex $u\in V(F)$. 
If $X \cap V(F) \neq \emptyset$ or there is no $S_i$-$\{r(F)\}$ path in $D-X$ we are done. Otherwise, $r(F)$ is reachable from $S_i$ in $D-X$. Since $r(F)$ is the root of $F$ and $F \subseteq D-X$, all the vertices of $V(F)$ are reachable from $S_i$ in $D-X$. For the sake of contradiction, suppose that for each $u \in V(F)$, there is a path from $u$ to $S_i$ in $D-X$. Then $F$ should be in the same connected component as $S_i$ in $D-X$. This contradicts that $X$ is a solution.

For the backward direction, suppose that there is a nice solution for $(D,\S=(S_1,\ldots,S_q), \allowbreak W,k)$. Since every nice solution is a solution for the instance $(D,W,k)$ of {\sc r-$\HH$-SCC DC}, we are done.
\end{proof}

Since the number of ordered partitions of a set $W$ is at most $\vert W \vert^{\vert W \vert}$ and the size of $W$ in an instance $(D,W,k)$ of {\sc r-$\HH$-SCC DC} is $k+1$, we conclude from Lemma~\ref{lem:DisjCompressionToDisjPartionedCompression} that to solve {\sc r-$\HH$-SCC DC}, it is enough to solve $2^{O(k \log k)}$ nice instances of {\disjPartitionedCompression}. Towards this, in order to prove Lemma~\ref{lem:disjCompressionSolution} (and hence Theorem~\ref{thm:rootedSCCDeletion}), it is enough to prove the following lemma.

\begin{restatable}{lemma}{disjPartitionedCompressionSolution}\label{lem:disjPartitionedCompressionSolution}
There is an algorithm that, given an instance $\mathcal{I}=(D,\S=(S_1,\dots,S_q),W,k)$ of {\disjPartitionedCompression}, runs in time $O^*(2^{O(k^3 \log k)})$ and either correctly concludes that there is no nice solution for $(D,\S,W,k)$ or outputs some solution for $(D,\S,W,k)$.
\end{restatable}

The upcoming subsections are devoted to the proof of Lemma~\ref{lem:disjPartitionedCompressionSolution}.
Recall that the most challenging aspect in our strategy (see overview of our algorithm for {\sc r-$\HH$ SCC Deletion} in Section~\ref{sec:intro}) is the identification of appropriate `branchable' forbidden subgraphs. Specifically, we need to identify particular subgraphs $F$ such that the natural exhaustive branchings reduce some measure depending on the parameter.
% and thus the algorithm terminates. 
 The way we identify such subgraphs is the following:  the algorithm will maintain a set $Q$ such that $Q$ is root-$\cal H$-free, with $Q=\emptyset$ initially. The algorithm will try to `grow' the set $Q$ until the entire graph is covered by $Q$. Initially, when $Q=\emptyset$, we try to grow it to the set $R_{\rm min}(S,T)$, the  closest minimum $S$-$T$ separator. We prove that if we find a forbidden subgraph whose root lies in $R_{\rm min}(S,T)$, that subgraph is good for us in the sense that all branches will drop our measure. Once all roots have been removed from $R_{\rm min}(S,T)$, we set $Q=R_{\rm min}(S,T)$. Then we recurse and grow $Q$ further towards $T$. To formalize the above strategy in Section~\ref{sec:rootedgeneral}, we next describe two crucial tools, in the form of pushing lemmas. 

\subsection{Setting up the Required Notations and Machinery}

Let $(D,\S=(S_1, \ldots, S_q),W,k)$ be an instance of {\sc r-$\HH$-SCC PC}. Let $X$ be some solution for this instance. Suppose, for instance, one had a hold on the set of vertices, say $R$, that are in the strong components containing $S_1$ in $D-X$. Then, one could argue that there is some other solution that picks an important $R$-$(W \setminus S_1)$ separator
% (see Marx' survey~\cite{marx2006parameterized} for definition) 
of size at most $k$. In this case, one can branch of these important sets. Unfortunately, the above mentioned set $R$ is far from being found efficiently. However growing on this idea, our first pushing lemma, Lemma~\ref{lem:pushingLemmaOne}, says that even if one is able to find a ``weaker'' set, viz. some superset of the vertices that are reachable from $S_1$ after the deletion of the solution, do not contain a graph of $\HH$ inside them and their out-neighbourhood forms a minimum $S_1$-$(W\setminus S_1)$,
then one can always construct another solution that picks the out-neighbours of this set.

\begin{lemma}\label{lem:soln_property}
	Let $X$ be an $\cH$-deletion set of $D$. Let $F \subseteq D$ be isomorphic to some graph in $\cH$. Suppose there exists a set of paths $\mathcal{P}$ containing a path $ P_{xy}$ from $x$ to $y$ in $D$ for each $(x,y) \in V(F) \times V(F)$ (equivalently, $F$ is contained in a single strong component of $D$).
Then there exists a vertex $a \in X$ such that for every $v \in V(F)$, there exists a path from $v$ to $a$ and a path from $a$ to $v$ that is contained in $\bigcup_{(x,y) \in V(F) \times V(F)} P_{xy}$.
\end{lemma}

\begin{proof}
Consider the collection $\mathcal{P}$ of paths described in the lemma. Since $X$ is an $\cH$-deletion set of $D$, $F$ is not contained inside one strongly connected component of $D-X$. Thus, there exists $x,y \in V(F)$, such that $P_{xy} \in \mathcal{P}$ is not entirely contained in $D-X$, that is, $V(P_{xy}) \cap X \neq \emptyset$. Let $a \in V(P_{xy}) \cap X$. Consider the subpath, say $P' \subseteq P_{xy}$, from $x$ to $a$. Since $a \in X$, $P'$ is the desired path from $x$ to $a$ in $D$ contained inside $\bigcup_{(x,y) \in V(F) \times V(F)} P_{xy}$. For any other $v \in V(F) \setminus \{x\}$, consider the path $P_{vx}$ described in the lemma. $P_{vx}$ together with $P'$ give a walk (and eventually a path) from $v$ to $a$ (and hence $X$) that is contained inside $\bigcup_{(x,y) \in V(F) \times V(F)} P_{xy}$.
A symmetric argument shows that we can also get a path from $a$ to $v$, for each $v \in V(F)$, that is contained inside $\bigcup_{(x,y) \in V(F) \times V(F)} P_{xy}$. 
\end{proof}

\begin{lemma}[{\sc Pushing-Routine-1}]\label{lem:pushingLemmaOne}
Let $\mathcal{I}=(D,\S=(S_1,\dots, S_q),W,k)$ be an instance of {\disjPartitionedCompression}. 
Consider a $\cal H$-free set $S_1\subseteq Q\subseteq V(D)\setminus (W\setminus S_1)$ such that $N^+(Q)$ is a minimum $S_1$-$(W\setminus S_1)$ separator. Let $X$ be a solution of $\mathcal{I}$ such that $R_D(S_1,X) \cap N^+(Q)=\emptyset$.
Then, there is a solution $X'$ for $\mathcal{I}$ that contains $N^+(Q)$. 
\end{lemma}

\begin{proof}
	Let $R=R_D(S_1,X)$. We will first show that $N^+(Q)$ covers $N^+(R)$. Observe that both $N^+(R)$ and $N^+(Q)$ are $S_1$-$(W \setminus S_1)$ separators in $D$. Also, since $R \cap (W \setminus S_1) =\emptyset$, $N^+(R)$ is a minimal $S_1$-$(W \setminus S_1)$ contained in $X$. Since $R_D(S_1,X) \cap N^+(Q)=\emptyset$, $R_D(S_1. R) \subseteq R_D(S_1,Q)$. Thus, $N^+(Q)$ covers $N^+(R)$.
	
We will now construct a set $X'$ that contains $N^+(Q)$, and prove that it is also a solution for $\mathcal{I}$. Let $X' = (X \setminus N^+(R)) \cup N^+(Q)$. Since both $N^+(R)$ and $N^+(Q)$ are $S_1$-$W \setminus S_1$ separators in $D$ and $N^+(Q)$ is a minimum one, $\vert N^+(Q) \vert \leq \vert N^+(R) \vert$. Thus, $\vert X' \vert \leq \vert X \vert$. Also, since $X \cap W =\emptyset$ and $N^+(Q) \cap W =\emptyset$ (because $N^+(Q)$ is a minimum $S_1$-$(W\setminus S_1)$ separator), we conclude that $X' \cap W = \emptyset$. Also observe that $R_D(S_1,X') \subseteq Q$.
	
Suppose that there exists $F \subseteq D-X'$ which is isomorphic to some graph in $\HH$ and which is contained in a single strongly connected component of $D-X'$. Then, for each $(x,y) \in V(F) \times V(F)$, there exists paths $P_{xy}$ that are entirely contained in $D-X'$.
We will first prove that for each $v \in V(F)$, $v \in \overline{R}_D(S_1, X')$. Suppose not. Let $v \in V(F)$ be such that $v \in R_D(S_1, X')$. Then since the paths $P_{x,y}$ (for each $(x,y) \in V(F) \times V(F)$) are in $D-X'$, this implies $V(F) \subseteq R_D(S_1,X') \subseteq Q$. This contradicts that $Q$ is $\HH$-free. 
Since, for all $(x,y) \in V(F) \times V(F)$, $x,y \in \overline{R}_D(S_1,X')$, $P_{xy}$ is a path from $x$ to $y$ in $D-X'$ and $N^+(Q)$ covers $N^+(R)$, from the construction of $X'$ and Lemma~\ref{lem:pushing_property}, for all $(x,y) \in V(F) \times V(F)$, there exists a path from $x$ to $y$ in $D-X$. This implies that $F$ is a strongly connected graph in $D-X$, and therefore present in some strongly connected component of $D-X$. This contradicts that $X$ is a solution to $\mathcal{I}$.

We will now show that for each $i,j \in [q]$, $i < j$, there is no path from $S_i$ to $S_j$ in $D-X'$. Since $N^+(Q) \subseteq X'$ and $N^+(Q)$ is an  $S_1$-$(W \setminus S_1)$ separator in $D$, there is no path from $S_1$ to $S_j$, $j >1$. Consider any pair $S_i,S_j$, $i <j$. Note that $S_i,S_j \subseteq W \setminus S_1$. Also, $S_i,S_j \in \overline{R}_D(S_1,X')$. For the sake of contradiction, suppose there exists a vertex $a \in S_i$ and $b \in S_j$ such that there is a path from $a$ to $b$ in $D-X'$. Since $a,b \in \overline{R}_D(S_1,X')$, and $N^+(Q)$ covers $N^+(R)$, from the construction of $X'$ and Lemma~\ref{lem:pushing_property}, we conclude that there is a path from $a$  to $b$, and hence from $S_i$ to $S_j$, in $D-X$ too. This contradicts that $X$ is a solution to $\mathcal{I}$.
\end{proof}

For our second pushing lemma, we first borrow definitions of shadows for directed graphs from~\cite{ChitnisHM13}. Note that what we do with the concept of shadows in our article is very different from the way it has been used so far. More specifically, we do not resort to the shadow removal technique that has often been an effective technique to design FPT algorithms for directed cut problems. In fact, for the general problems that we consider, it is not at all clear how the shadow removal  technique could be helpful.
Let $(D,\S=(S_1,\dots, S_q),\allowbreak W, \allowbreak k)$ be an instance of {\disjPartitionedCompression} and let $X\subseteq V(D)\setminus W$. Then, $\forwardshadow(X)$ denotes the set of those vertices $u\notin X$ such that there is no  $\{u\} $-$W$ path in $D$-$X$. Similarly, $\reverseshadow(X)$ denotes the set of those vertices $u\notin X$ such that there is no  $W$-$\{u\} $ path in $D$-$X$.
%\end{definition}
$\forwardshadow(X)$ is called the {\em forward shadow of $X$ with respect to $W$} and $\reverseshadow(X)$ is called the {\em reverse shadow of $X$ with respect to $W$}. 
Notice that with these definitions, we have that if $X$ is a nice solution for the instance $(D,\S=(S_1,\dots, S_q),W,k)$, 
then for any subgraph $F \subseteq D$ that is isomorphic to a graph in $\HH$ and any $i \in [q]$, either $X$ intersects $V(F)$ or $r(F)\notin R(S_i,X)$ 
or $V(F)\cap \forwardshadow(X)\neq \emptyset$. Moreover, when $q=1$, this implies that $X$ intersects $V(F)$ or $r(F)\in \reverseshadow(X)$ or $V(F)\cap \forwardshadow(X)\neq \emptyset$.
Our second pushing lemma, guarantees the existence of a set to branch on, provided we have identified some vertex in the forward or reverse shadow of $X$ with respect to $W$.
For a digraph $D$, $D^{rev}$ denotes the digraph obtained by reversing the direction of all arcs in $D$.

\begin{lemma}
\label{lem:pushingLemmaTwo}
Let $\mathcal{I}=(D,\S=(S_1,\dots, S_q),W,k)$ be an instance of {\disjPartitionedCompression} and 
let  $X$ be a solution for this instance. 
Let $u\in V(D)$ be a vertex in $\forwardshadow(X)$ (resp.~$\reverseshadow(X)$). Then, there is a solution for $\mathcal{I}$ that contains an important $\{u\}$-$W$ separator in $D$ (resp.~an important $\{u\}$-$W$ separator in $D^{rev}$). 
\end{lemma}

\begin{proof}
	We prove the case when $u \in \forwardshadow(X)$. The case when $u \in \reverseshadow(X)$ is symmetrical.
%	$u$ is in the forward shadow, the case for reverse shadow follows by a symmetrical argument. 
Since $u$ is in the forward shadow of $X$ w.r.t. $W$, there is no $\{u\}$-$W$ path in $D-X$ and thus, $X$ is a $\{u\}$-$W$ separator. Let $R=R_D(\{u\},X)$. Since $R$ is exactly the set of vertices reachable from $u$ after removing $X$, it follows that $N^+(R) \subseteq X$. Moreover, since $u \in R$ and $R \cap W = \emptyset$ it follows that $N^+(R)$ is a $\{u\}$-$W$ separator. Thus there exists an important $\{u\}$-$W$ separator $C$ that covers $N^+(R)$ such that $\vert C \vert 
\leq \vert N^+(R) \vert$. Let $X' = (X \setminus N^+(R)) \cup C$. Clearly $\vert X' \vert \leq \vert X \vert$. Also, since $X \cap W = \emptyset$ and $C \cap W =\emptyset$ (because $C$ is a $\{u\}$-$W$ separator), $X' \cap W = \emptyset$. We now prove that that $X'$ is a solution for $\mathcal{I}$.

Suppose that there exists $F \subseteq D-X'$ which is isomorphic to some graph in $\HH$ and which is contained in a single strongly connected component of $D-X'$. Then, for each $(x,y) \in V(F) \times V(F)$, there exists paths $P_{xy}$ that are entirely contained in $D-X'$.
We will first prove that for each $v \in V(F)$, $v \in \overline{R}_D(\{u\}, X')$. Suppose not. Let $v \in V(F)$ be such that $v \in R_D(\{u\}, X')$.
In this case, first note that $v \not \in W$ because $u \in \forwardshadow(X)$. 
Since, $W$ is a solution to $D$ for the $\HH$-SCC problem, from Lemma~\ref{lem:soln_property}, there exists a path from $v$ to $W$ in $\bigcup_{(x,y) \in V(F) \times V(F)} P_{xy}$. Since the paths $P_{xy}$ are in $D-X'$, we conclude that there is a path from $v$ to $W$ in $D-X'$. Also, since $v \in R_D(\{u\},W)$, we conclude that there is a path from $u$ to $W$ in $D-X'$, thereby contradicting that $u \in \forwardshadow(X)$. 
Thus, we now have that for all $(x,y) \in V(F) \times V(F)$, $x,y \in \overline{R}_D(\{u\}, X')$. Since $C$ covers $N^+(R)$ and there is a path $P_{x,y}$ from $x$ to $y$ in $D-X'$, from the construction of $X'$ and Lemma~\ref{lem:pushing_property}, there is also a path from $x$ to $y$ in $D-X$. This implies that $F$ is a strongly connected in $D-X$ and hence is contained in some strongly connected component of $D-X$. This contradicts that $X$ is a solution to $\mathcal{I}$.

We will now show that for each $i,j \in [q]$, $i < j$, there is no path from $S_i$ to $S_j$ in $D-X'$. 
Consider any pair $S_i,S_j$, $i <j$. For the sake of contradiction, suppose there exists a vertex $a \in S_i$ and $b \in S_j$ such that there is a path from $a$ to $b$ in $D-X'$. Since $S_i,S_j \subseteq W$ and $u \in \forwardshadow(X)$, $a,b \in \overline{R}_D(\{u\},X')$. Since $C$ covers $N^+(R)$, from the construction of $X'$ and Lemma~\ref{lem:pushing_property}, there is a path from $a$ to $b$, and hence from $S_i$ to $S_j$, in $D-X$. This contradicts that $X$ is a solution to $\mathcal{I}$.
\end{proof}

As a consequence of Lemma~\ref{lem:pushingLemmaTwo} and Lemma~\ref{lem:imp_sep_bound}, we have the following.

\begin{lemma}[{\sc Pushing-Routine-2}]
	\label{lem:pushingRoutineTwo}
There is an algorithm that, given an instance $\mathcal{I}=(D,\S,W, \allowbreak k)$ of {\disjPartitionedCompression} and a vertex $u\in V(D)$ such that either there is a $u$-$W$ path or a $W$-$u$ path in $D$, 
 runs in time $4^k \cdot n^{O(1)}$ and outputs a non-empty set $Z \subseteq V(D)$ of size at most $4^k \cdot 2k$ with the following property: if there is a solution $X$ for $\mathcal{I}$ such that $u \in \forwardshadow(X) \cup \reverseshadow(X)$, then there is a solution $X'$ for $\mathcal{I}$ such that $X' \cap Z \neq \emptyset$. 
\end{lemma}

\begin{proof}
Given the instance $\mathcal{I}$ and $u \in V(D)$, let $\mathcal{F}_1$ be the family of important $\{u\}$-$W$ separators of size at most $k$ in $D$ and let $\mathcal{F}_2$ be the family of important $\{u\}$-$W$ separators of size at most $k$ in $D^{rev}$. Let $Z = \cup_{X \in \mathcal{F}_1 \cup \mathcal{F}_2} X$. Note that since there is either a $u$-$W$ path or a $W$-$u$ path, $Z$ must be non-empty. Then from Lemma~\ref{lem:pushingLemmaTwo}, there is a solution $X'$ of $\mathcal{I}$ such that $X' \cap Z \neq \emptyset$. Also, from Lemma~\ref{lem:imp_sep_bound}, $\vert Z \vert \leq 4^k \cdot 2k$.
\end{proof}

\subsection{Solving instances of {\disjPartitionedCompression
} with a trivial partition}\label{sec:rootedtrivial}

\newcommand{\fbs}{{\sc Find-Branch-Set}}

Towards the proof of Lemma~\ref{lem:disjPartitionedCompressionSolution}, we first find a set~$\widehat{Z}$ such that~$\widehat{Z}$ intersects some solution for $\mathcal{I} =(D,\S,W,k)$ (if one exists) and $\vert \widehat{Z} \vert = O(h) \cdot 2^{O(k^2  \log k)}$. Observe that, having such a set~$\widehat{Z}$ at hand, one can proceed with a branching algorithm that branches on the vertices of~$\widehat{Z}$. We call the set~$\widehat{Z}$, the {\em branch set} for the instance~$\mathcal{I}$.
The rest of the section is devoted to computing a branch set for $\mathcal{I}$ of the desired size. 

\begin{lemma}\label{lem:findBranchSet}
Given an instance $\mathcal{I}=(D,\S,W,k)$ of {\sc r-$\HH$-SCC PC}, there is an algorithm, that runs in time $2^{O(k^2 \log k)} \cdot n^{O(h)}$ and outputs a branch set for $\mathcal{I}$ of size  $\O(h) \cdot 2^{O(k^2 \log k)}$ if $\mathcal{I}$
has a nice solution.
\end{lemma}

The algorithm of Lemma~\ref{lem:findBranchSet} has two parts: in this section we design a simple algorithm when $q=1$ by exploiting the structure of a nice solution to identify a vertex that belongs to the shadow of the solution, thereby allowing the applicability of Lemma~\ref{lem:pushingRoutineTwo} to find a branch set. 
%We defer the details of this part to the full version.
 The second part, when $q>1$, is tricky. We design a recursive algorithm for the proof of Lemma~\ref{lem:findBranchSet} when $q>1$ in Section~\ref{sec:rootedgeneral}.

 %, whose details we describe next. To maintain the recursive invariants, we enhance the instance of {\sc r-$\HH$-SCC PC}.

\begin{lemma}\label{lem:baseCasedisjPartitionedCompressionSolution}
There is an algorithm that, given an instance $\mathcal{I}=(D,\S=(S_1),W,k)$ of {\disjPartitionedCompression}, runs in time $2^{O(k^2)} \cdot n^{O(h)}$ and either correctly concludes that there is no nice solution for $\mathcal I$ or outputs a solution for $\mathcal I$.
% \red{The algorithm just has to output a solution for $(D,k)$. We don't care if the output is nice or not.}
\end{lemma}

% \todo[inline]{See if you really want to use the first line of the proof-- The algorithm, which we call \textsc{BaseCase}... -- instead one can give the name to the lemma as is done in previous lemmas.§}

% \todo[inline]{Why do you need the arguments in red below.}
\begin{proof}
	The algorithm proceeds by checking if there is any graph, say $F \subseteq D$ such that $F$ is isomorphic to some graph in $\HH$ and is contained in a single strong component of $D$. If there is no such graph, then $\emptyset$ is a solution to $\mathcal{I}$. Otherwise, let $F$ be such a graph. Suppose there exists a nice solution, say $X$, for $\mathcal{I}$. Then from the definition of a nice solution, $X$ either contains some vertex of $F$, or contains an $W$-$\{r(F)\}$ separator, in which case $r(F) \in \reverseshadow(X)$, or contains an $\{v\}$-$W$ separator for some $v \in V(F)$, in which case $v \in \forwardshadow(X)$. By Lemma~\ref{lem:soln_property}, there is a $v$-$W$ and $W$-$v$ path for every $v \in V(F)$, thus we can call Lemma~\ref{lem:pushingRoutineTwo} on vertex $v$. For each $v \in V(F)$, let $Z_{v}$ be the set outputted by the algorithm of Lemma~\ref{lem:pushingRoutineTwo} when given input $(\mathcal{I},v)$. %{\color{red} Note that the set $Z$ will never be empty since by \cref{lem:soln_property}, there exists a vertex $a \in W$ such that there is a $a$-$x$ and a $x$-$a$ path for every $x \in V(F)$ (here the set of paths $\cal P$ required to invoke the statement of \cref{lem:soln_property} are the paths witnessing the fact that $F$ is contained in a strong component of $D$).} 
	Let $Z= V(F) \cup \bigcup_{v \in V(F)} Z_v$. Note that $\vert Z \vert \leq h + h (4^k \cdot 2k)$. From the arguments above and Lemma~\ref{lem:pushingRoutineTwo}, there exists a solution that contains some vertex of $Z$. The algorithm branches on the set $Z$, that is, for each $v \in Z$, the algorithm creates an instance $(D-\{v\}, \S, W,k-1)$.
	The branching algorithm stops when either $k < 0$ (in which case we conclude that it is a no-instance) or when the there is no subgraph isomorphic to a graph in $\HH$ (in which case it is a yes-instance).
	Since, if there exists a nice solution, then there exists a solution that contains some vertex of $Z$, the correctness of the branching algorithm follows. For the running time analysis, since the branching factor is at most $h + h (4^k \cdot 2k)$ and the branching depth is at most $k$, we get the desired running time.
\end{proof}

\subsection{Solving general instances of {\disjPartitionedCompression
}}\label{sec:rootedgeneral}

We now design a recursive algorithm for Lemma~\ref{lem:findBranchSet} for the case when $q>1$.
To maintain the recursive invariants, we enhance the instance of {\sc r-$\HH$-SCC PC}. 
%Henceforth, $q>1$.

\begin{definition}[Extended Instance of {\sc r-$\HH$-SCC PC}]\label{defn:extdefn}
An instance $\mathcal{I}^{\rm ext} = (D,\S=(S_1, \ldots, S_q),W,\allowbreak k, \allowbreak S,T,Q,N_Q)$ is called an {\em extended instance of {\sc r-$\HH$-SCC PC}} if the following holds: 
\begin{enumerate}
\item $(D,\S=(S_1, \ldots, S_q),W,k)$ is an instance of {\sc r-$\HH$-SCC PC},

\item $S_1 \subseteq S \subseteq V(D) \setminus (W \setminus S_1)$ and $W \setminus S_1 \subseteq T $,

\item  either $Q = \emptyset$ or, $S \subseteq Q \subseteq V(D) \setminus T$ such that $Q$ is root-$\HH$-free and $N^+(Q)$ is a minimum $S$-$T$ separator in $D$, and

\item $N_Q \subseteq N^+(Q)$.
\end{enumerate}

 \end{definition}
 
 \begin{definition}[Solution of an extended instance of {\sc r-$\HH$-SCC PC}]\label{defn:extsol}
For an extended instance $\mathcal{I}^{\rm ext} = (D,\S=(S_1, \ldots, S_q),W,k,S,T,Q,N_Q)$ of {\sc r-$\HH$-SCC PC}, $X \subseteq V(D) \setminus W$ is said to be a solution for $\mathcal{I}^{\rm ext}$ if the following holds.
\begin{enumerate}
\item $X$ is a nice solution for $(D,\S,W,k)$, 

\item $X$ is an $S$-$T$ separator, 

\item $S \subseteq R_D(S_1,X)$,

\item $N_Q \cap R_D(S,X) =\emptyset$ and,

\item $X \cap (S \cup T) = \emptyset$.
\end{enumerate}
%{\rm (1)} $X$ is a nice solution for $(D,\S,W,k)$, {\rm (2)} $X$ is an $S$-$T$ separator, {\rm (3)} $S \subseteq R_D(S_1,X)$, {\rm (4)} $N_Q \cap R_D(S,X) =\emptyset$ and, {\rm (5)} $X \cap (S \cup T) = \emptyset$.
 \end{definition}

The idea behind extending an instance of {\sc r-$\HH$-SCC PC} in the way defined earlier is to get the situation closer to the applicability of Lemma~\ref{lem:pushingLemmaOne}. In fact, as we will see, this will form the base case for our arguments.
To be more specific, the sets $S,T $ defined in the definition correspond to the set of vertices that one has guessed to be reachable and unreachable, respectively from $S_1$ in $D-X$, where $X$ is a solution for the original instance. Thus, any solution for the original instance is an $S$-$T$ separator. Note that, since $X$ is a solution for $(D,\S,W,k)$, the set $S_1$ itself is reachable from $S_1$ and $W\setminus S_1$ is unreachable from $S_1$ in $D-X$. Therefore we could assume that $S_1 \subseteq S$ and $(W\setminus S_1) \subseteq T$. The set $Q$ in the extended instance is such that $N^+(Q)$ is a minimum $S$-$T$ separator and $Q$ itself is root-$\HH$-free (and hence $\HH$-free). The set $N_Q$ is meant to be the subset of $N^+(Q)$ that one has guessed to be unreachable from $S$ in $D-X$.
The algorithm aims to slowly ``grow'' $Q$ using Lemma~\ref{lem:tightSuccessorComputation} until we find a subgraph $F$ in $D$ that is isomorphic to some graph in $\HH$ and whose root lies in $Q$ (i.e. $Q$ is no longer root-$\cal H$-free). Then using the fact that a nice solution exists, one can construct the desired branch set by branching of instances with a smaller, appropriately defined, measure.

\begin{observation}\label{obs:extendedsoln}
If $(D, \S,W,k)$ has a nice solution to the problem {\sc r-$\HH$-SCC PC}, then $(D,\S,W, \allowbreak k, \allowbreak S_1,W\setminus S_1,\emptyset,\emptyset)$ has a solution.
\end{observation}

Observe that, from Observation~\ref{obs:extendedsoln}, Lemma~\ref{lem:extbranchset} implies Lemma~\ref{lem:findBranchSet}.

\begin{lemma}\label{lem:extbranchset}
%[{\sc Find-Branch-Set}]
Given an extended instance $\mathcal{I}^{ext} = (D,\S=(S_1, \ldots, S_q), \allowbreak W, \allowbreak k,S,T,Q,N_Q)$ of {\sc r-$\HH$-SCC PC}, there is an algorithm that runs in time $2^{O(k^2 log k)} \cdot n^{O(h)}$, and returns a branch set of $(D,\S,W,k)$ of size $O(h) \cdot 2^{O(k^2 \log k)}$, if $\mathcal{I}^{ext}$ has a solution.
\end{lemma}

For the convenience of arguments, we further enhance an extended instance of {\sc r-$\HH$-SCC PC}. The idea behind this  is to avoid asking that $N_Q$ has to be unreachable from $S$ in $D-X$ where $X$ is a solution of the extended instance. As we see below, a slight modification to the digraph $D$ and $T$ help us achieve this, thereby easing the arguments used in the final proof.

\begin{definition}[Auxiliary Instance of {\sc r-$\HH$-SCC PC}]
Given an extended instance $\mathcal{I}^{\rm ext} = (D,\S=(S_1, \ldots, S_q), \allowbreak W, \allowbreak k,S,T,Q,N_Q)$ of {\sc r-$\HH$-SCC PC}, we define an auxiliary instance $\mathcal{I}^{\rm aux} = (D^{\rm aux},\S=(S_1, \ldots, S_q), \allowbreak W, \allowbreak k,S,T^{\rm aux},Q,N_Q)$ of {\sc r-$\HH$-SCC PC} as follows: $D^{\rm aux}$ is a supergraph of $D$ that is obtained from $D$ by adding a new vertex $t^{\rm aux}$ in $D$ and adding arcs $(u,t^{\rm aux})$, for each $u \in N_Q$; $T^{\rm aux}= T \cup \{t^{\rm aux}\}$.
\end{definition}

\begin{definition}[Solution of an auxiliary instance of {\sc r-$\HH$-SCC PC}]
Let $\mathcal{I}^{\rm aux} = (D^{\rm aux},\S=(S_1, \ldots, S_q),W,k,S,T^{\rm aux},Q,N_Q)$ be an auxiliary instance of {\sc r-$\HH$-SCC PC} obtained from $\mathcal{I}^{\rm ext} = (D,\S=(S_1, \ldots, S_q), \allowbreak W, \allowbreak k,S,T,Q,N_Q)$, then $X \subseteq V(D) \setminus W$ is said to be a solution for $\mathcal{I}^{\rm aux}$ if the following holds: 
\begin{enumerate}
\item  $X$ is a nice solution for $(D,\S,W,k)$,

\item $X$ is an $S$-$T^{\rm aux}$ separator in $D ^{\rm aux}$, 

\item $S \subseteq R_{D^{\rm aux}}(S_1,X)$ and,

\item $X \cap (S \cup T^{\rm aux})=\emptyset$.
\end{enumerate}
%{\rm (1)} $X$ is a nice solution for $(D,\S,W,k)$,  {\rm (2)} $X$ is an $S$-$T^{aux}$ separator in $D ^{aux}$, {\rm (3)} $S \subseteq R_{D^{aux}}(S_1,X)$ and, {\rm (4)} $X \cap (S \cup T^{aux})=\emptyset$.
 \end{definition}
 
 We will now show that a solution to an extended instance of {\sc r-$\cal H$-SCC PC} is also a solution for the corresponding auxiliary version. 
 
 \begin{lemma}\label{lem:extToAux}
Let $\mathcal{I}^{aux} = (D^{aux},\S=(S_1, \ldots, S_q),W,k,S,T^{aux},Q,N_Q)$ be an auxiliary instance of {\sc r-$\HH$-SCC PC} obtained from $\mathcal{I}^{ext} = (D,\S=(S_1, \ldots, S_q), \allowbreak W, \allowbreak k,S,T,Q,N_Q)$. If $X$ is a solution for $\mathcal{I}^{ext}$ then $X$ is also a solution for $\mathcal{I}^{aux}$.% If $X$ is a solution for $\mathcal{I}^{ext}$ then $X$ is a solution for $\mathcal{I}^{aux}$. 
\end{lemma}
\begin{proof}
	Let $X$ be a solution for $\mathcal{I}^{ext}$. Then $X$ is already a nice solution for $(D,\S,W,k)$ and $X \cap (S \cup T^{aux})=\emptyset$ follows trivially. Since $t^{aux}$ is a vertex with no out-neighbours, there cannot be any path that pass through $t^{aux}$ and have endpoints not equal to $t^{aux}$, thus there cannot be any $S-T$ paths passing through $t^{aux}$ and no $S_1-\{v\}$ paths through $t^{aux}$ for any vetex $v \not\in R_D(S_1,X)$, it follows that $X$ is an $S-T$ cut in $D^{aux}$ and $S \subseteq R_D(S_1,X) \subseteq R_{D^{aux}}(S_1,X)$. Moreover, since $N_Q$ is unreachable from $S$ in $D-X$, it follows that $t^{aux}$ is unreachable from $S$ in $D-X$ and thus $X$ is a $S-T^{aux}$ separator in $D^{aux}$. 
\end{proof}

%\begin{lemma}\label{lem:extToAux}
%Let $\mathcal{I}^{aux} = (D^{aux},\S=(S_1, \ldots, S_q),W,k,S,T^{aux},Q,N_Q)$ be an auxiliary instance of {\sc r-$\HH$-SCC PC} obtained from $\mathcal{I}^{ext} = (D,\S=(S_1, \ldots, S_q), \allowbreak W, \allowbreak k,S,T,Q,N_Q)$. If $X$ is a solution for $\mathcal{I}^{ext}$ then $X$ is also a solution for $\mathcal{I}^{aux}$.% If $X$ is a solution for $\mathcal{I}^{ext}$ then $X$ is a solution for $\mathcal{I}^{aux}$. 
%\end{lemma}
%\vspace{-5 pt}

%Henceforth, we will denote an auxiliary instance of {\sc r-$\HH$-SCC PC} as $\mathcal{I}^{aux} = (D,\S=(S_1, \ldots, S_q), \allowbreak W, \allowbreak k,S,T,Q,N_Q)$.
%Lemmas~\ref{lem:extToAux} and~\ref{lem:auxbranchset}, together with the deferred proof of Lemma~\ref{lem:findBranchSet} for the case when $q=1$, implies Lemma~\ref{lem:findBranchSet}.

From Lemma~\ref{lem:extToAux}, Lemma~\ref{lem:auxbranchset} implies Lemma~\ref{lem:extbranchset}.

\begin{lemma}[{\sc Find-Branch-Set}]\label{lem:auxbranchset}
%[{\sc Find-Branch-Set}]
Given an auxiliary instance $\mathcal{I}^{\rm aux} = (D^{\rm aux},\S=(S_1, \ldots, S_q), \allowbreak W, \allowbreak k,S,T^{\rm aux},Q,N_Q)$ of {\sc r-$\HH$-SCC PC} where $q>1$, there is an algorithm that runs in time $2^{O(k^2 log k)} \cdot n^{O(h)}$ and returns a branch set of $(D,\S,W,k)$ of size $O(h) \cdot 2^{O(k^2 \log k)}$ if $\mathcal{I}^{\rm aux}$ has a solution.
\end{lemma}
\vspace{-5 pt}

\begin{proof}
Let $\mathcal{I}=(D,\S,W,k)$. We will design a recursive algorithm {\sc Find-Branch-Set}. To analyse the depth of recursion, we associate a measure $\mu$ with an instance $\mathcal{I}^{\rm aux} = (D^{\rm aux},\S=(S_1, \ldots, S_q), \allowbreak W, \allowbreak k,S,T^{\rm aux},Q,N_Q)$ as $\mu(\mathcal{I}^{\rm ext}) =k^2 + k -\lambda_{D^{\rm aux}}(S,T^{\rm aux})^2 - \vert N_Q \vert$. For the sake of convenience, we will denote $\lambda_{D^{\rm aux}}(S,T^{\rm aux})$ by $\lambda(\mathcal{I}^{\rm aux})$. In what follows, we give an exhaustive list of cases, and say what the algorithm outputs in each such case, give a proof of correctness for the same, point out the branching width of the recursive calls and argue that the measure $\mu$ drops for each of the instances called in each of the recursive calls.
% and finally, argue about the size of the set returned.

\bigskip
\noindent{\bf Base Case:} Observe that for any auxiliary instance $\mathcal{I}^{\rm aux}$, if $\lambda(I^{\rm aux}) > k$, then $\mathcal{I}_{\rm aux}$ has no solution. If $k\leq 0$, then check if any strong   component of $D$ has a graph isomorphic to some graph in $\HH$. If it does, then $\mathcal{I}^{\rm aux}$ has no solution, otherwise, return $\emptyset$. Another case that is handled as a base case is when either $\mu(\mathcal{I}^{\rm aux}) \leq 0$ or $\vert N_Q \vert =\lambda(\mathcal{I}^{\rm aux})$.
%	{\noindent \bf Case 1: [$\mu(\mathcal{I}^{aux}) \leq 0$ or $\vert N_Q \vert =\lambda(\mathcal{I}^{aux})$]: } 
If $\mu(\mathcal{I}^{\rm aux}) \leq 0$, we first claim that $\vert N_Q \vert = \vert N^+(Q) \vert = \lambda(\mathcal{I}^{\rm aux})$. Since $\vert N^+(Q) \vert = \lambda(\mathcal{I}^{\rm aux})$, it is enough to prove that $\vert N_Q \vert =\lambda(\mathcal{I}^{\rm aux})$. Since $N_Q \subseteq N^+(Q)$, we have that $\vert N_Q \vert \leq \lambda(\mathcal{I}^{\rm aux})$. For the sake of contradiction, suppose that $\vert N_Q \vert < \lambda(\mathcal{I}^{\rm aux})$. Since $\mu(\lambda(\mathcal{I}^{\rm aux})) \leq 0$, we have that $k^2 + k \leq \lambda(\mathcal{I}^{\rm aux})^2 + \lambda(\mathcal{I}^{\rm aux})$. This implies that $\vert N_Q \vert \geq \lambda(\mathcal{I}^{\rm aux})$, which is a contradiction. Thus, we have that $\vert N_Q \vert =\lambda(\mathcal{I}^{\rm aux})$. 
In this case, {\sc Find-Branch-Set$(\mathcal{I}^{\rm aux})$} returns $N^+(Q)$. 
%{\noindent \em Correctness: }
We now prove that $N^+(Q)$ is indeed a branch set for $\mathcal{I}^{\rm aux}$. First observe that, in this case $N_Q = N^+(Q)$. From the construction of $D^{\rm aux}$, there is an arc $(u,t^{\rm aux})$ for each $u \in N_Q$. 
Thus, in any solution $X$ of $\mathcal{I}^{\rm aux}$, $N_Q \cap R_{D^{\rm aux}}(S,X) = \emptyset$, that is, 
$N^+(Q) \cap R_{D^{\rm aux}}(S,X) =\emptyset$. Since any solution $X$ of $\mathcal{I}^{\rm aux}$ is also a solution for $\mathcal{I}$, we have that there exists a solution $X$ to $\mathcal{I}$, such that $N^+(Q) \cap R_{D^{\rm aux}}(S,X) = \emptyset$. Then, observe that $\mathcal{I}, Q,N^+(Q),X$ satisfy the properties of Lemma~\ref{lem:pushingLemmaOne}, and hence there exists a solution to $\mathcal{I}$ that contains a vertex of $N^+(Q)$. That is, $N^+(Q)$ is a branch set for $\mathcal{I}$.
%{\noindent \em Measure drop and Branching width: }This is irrelevant in this case, as we do not recurse (this is our base case).
Note that the size of the set outputted in this case is $\lambda(\mathcal{I}^{\rm aux}) \leq k$.
%{\noindent \em Size of output: } $\lambda(\mathcal{I}^{aux}) \leq k$.
%\todo[inline]{check if we assume that at the application of Case i we assume none of the previous cases are applicable?}
We now proceed to the recursive cases.

\bigskip
{\noindent \bf Case 1: [$Q=\emptyset$]: } The algorithm first computes the unique minimum closest $S$-$T^{\rm aux}$ separator $C$. It then checks if there exists a subgraph $F \subseteq D$ such that $F$ is isomorphic to some graph in $\HH$ and $r(F) \in R_D(S,C)$. 

%\todo[inline]{Break into cases.}

%\todo[inline]{This case needs explanation}

\medskip
{\noindent \bf Case 1.1: [$\nexists$ a subgraph $F \subseteq D$ such that $F$ is isomorphic to a graph in $\HH$ and $r(F) \in R_D(S,C)$]: } In this case, the algorithm returns {\sc Find-Branch-Set}$(D^{\rm aux},\S,W, \allowbreak k, \allowbreak S, \allowbreak T^{\rm aux}, R_D(S,C), \emptyset)$. 

{\noindent \em Correctness:} Note that in this case $Q$ is root-$\HH$-free and $N^+(Q)$ is a minimum $S$-$T^{\rm aux}$ separator in $D^{\rm aux}$. Thus, $X$ is a solution for the auxiliary instance $(D^{\rm aux},\S,W,k,S,T^{\rm aux}, \allowbreak R_D(S,C), \allowbreak \emptyset)$.

{\noindent \em Branching width: }It is $1$. 

{\noindent \em Measure drop: }Let $\mathcal{I}'^{\rm aux} = (D^{\rm aux},\S,W,k,S,T, R_D(S,C), \emptyset)$. Since the branching width is~$1$, in this case it is enough to prove that this case cannot occur more than $n$ times and the measure does not increase whenever this case arises. 
%The former claim follows because if $T^{\rm aux} \setminus t^{\rm aux}$ becomes equal to $V(D) \setminus S_1$, then there is a no solution to the problem and we stop. In fact, 
Since in the new instance $\mathcal{I}'^{\rm aux}$, the size of $Q$ has strictly increased, as it was an empty set before, the resulting instance~$\mathcal{I}'^{\rm aux}$ does not fall into this case again (as we will see later that in all the cases the set $Q$ only grows).
Since $k, N_Q$ remains the same in both the instances, and $\lambda(\mathcal{I}'^{\rm aux}) \geq \lambda(\mathcal{I}^{\rm aux})$ because $T^{\rm aux} \cup \{r(F)\} \supseteq T^{\rm aux}$, we conclude that $\mu(\mathcal{I}'^{\rm aux}) \leq \mu(\mathcal{I}^{\rm aux})$.

\medskip
{\noindent \bf Case 1.2: [$\exists F \subseteq D$ that is isomorphic to a graph in $\HH$, and $r(F) \in R_D(S,C)$]:}
%\todo[inline]{Name Lemma 21 as {\sc Pushing-Routine-2}}
In this case, the algorithm returns $V(F)$ $ \cup$ {\sc Find-Branch-Set}$(D^{\rm aux},\S,W,k,S,T^{\rm aux} \cup \{r(F)\} ,\emptyset,\emptyset)$ $ \cup$   $\bigcup_{u \in V(F)}$ {\sc Pushing-Routine-2}$(D,\S,W,k, \allowbreak u)$.

{\noindent \em Correctness: } Let $X$ be a solution for $\mathcal{I}^{\rm aux}$. Since by definition, $X$ is a nice solution of $\mathcal{I}$, either $V(F) \cap X \neq \emptyset$, or $X$ contains an $S_1$-$\{r(F)\}$ separator or if it does not satisfy either of the above two conditions, then it must contain a $\{u\}$-$S_1$ separator, for some $u \in V(F)$. In the first case, since the returned set contains $V(F)$, we are done. In the second case, if~$X$ contains an $S_1$-$\{r(F)\}$ separator and $S \subseteq R_D(S_1,X)$ (from the definition of an auxiliary solution), then $X$ must also contain an $S$-$\{r(F)\}$ separator. Thus, in this case $X$ must also be a solution to {\sc Find-Branch-Set}$(D^{\rm aux},\S,W,k,S,T^{\rm aux} \cup \{r(F)\} ,\emptyset,\emptyset)$. From induction hypothesis, $(D^{\rm aux},\S,W,k,S,T^{\rm aux} \cup \{r(F)\} ,\emptyset,\emptyset)$ returns a branch set for $(D,\S,W,k)$, and we are done. In the last case, if neither of the above two cases hold, then $X$ contains a $\{u\}$-$S_1$ separator for some $u \in V(F)$. We will prove that $u$ in in the forward shadow of $X$ with respect to $W$. Suppose, for the sake of contradiction, that there is a path from $u$ to a vertex $v \in W$ in $D-X$. Note that $v \not\in S_1$ as $X$ contains a $\{u\}-S_1$ separator. Since $r(F)$ is reachable from $S_1$ in $D-X$, $X$ is disjoint from $V(F)$, and there is a path from $r(F)$ to $u$ contained in $V(F)$ as $F$ is rooted, it follows that $u$ is reachable from $S_1$ in $D-X$. Furthermore, there is a path from $u$ to $v$ in $D-X$ and thus $v \in W \setminus S_1$ is reachable from $S_1$ in $D-X$. This is a contradiction to the fact that $X$ is a solution for $\mathcal{I}$. Therefore, it follows that $u$ is in the forward shadow of $X$ with respect to $W$. Moreover, since there is an $S_1$-$\{r(F)\} $ path in $D$, by rooted-ness there is an $S_1$-$u$ path and since $S_1 \subseteq W$, there is a $W$-$u$ path in $D$. By Lemma~\ref{lem:pushingRoutineTwo}, {\sc Pushing-Routine-2}$(D,\S,W,k,u)$ returns a branch set for $\mathcal{I}$. 

{\noindent \em Branching width: } It is $1$.

{\noindent \em Measure drop: } Let $\mathcal{I}'^{\rm aux} = (D^{\rm aux},\S,W,k,S,T^{\rm aux} \cup \{r(F)\} ,\emptyset,\emptyset)$. Since the branching width is $1$ and in the new instance $\mathcal{I}'^{\rm aux}$ the size of $T^{\rm aux}$ grows, this case does not happen more than $n$ number of times
%, where $n$ is the number of vertices of the digraph. This is 
because, if $T^{\rm aux} \setminus t^{\rm aux}$ becomes equal to $V(D) \setminus S_1$, then there is a no solution to the problem and we stop. Again, since the branching width is $1$, in this case, we only need to show that $\mu(\mathcal{I}'^{\rm aux}) \leq \mu(\mathcal{I}^{\rm aux})$. Since $k, N_Q$ remains the same in both the instances, and $\lambda(\mathcal{I}'^{\rm aux}) \geq \lambda(\mathcal{I}^{\rm aux})$ because $T^{\rm aux} \cup \{r(F)\} \supseteq T^{\rm aux}$, we conclude that $\mu(\mathcal{I}'^{\rm aux}) \leq \mu(\mathcal{I}^{\rm aux})$.

\bigskip
{\noindent \bf Case 2: [$Q\neq \emptyset$]: } The algorithm proceeds by invoking Lemma~\ref{lem:tightSuccessorComputation} and either finds a separator~$C'$ that tightly covers $N^+(Q)$ in $D^{\rm aux}$ or concludes that $N^+(Q)$ is the furthest minimum $S$-$T^{\rm aux}$ separator in $D^{\rm aux}$.

\medskip
{\noindent \bf Case 2.1: [$N^+(Q)$ is the furthest minimum $S$-$T^{\rm aux}$ separator in $D^{\rm aux}$]: } In this case, the algorithm returns $N^+(Q)$ $\cup$ $\bigcup_{v \in N^+(Q) \setminus N_Q}$ {\sc Find-Branch-Set}$(D^{\rm aux},\S,W,k, \allowbreak S \cup \{v\},T^{\rm aux},\emptyset,\emptyset)$.
	
{\noindent \em Correctness: }Let $\mathcal{I}^{\rm aux}_v=(D^{\rm aux},\S,W,k,S \cup \{v\},T^{\rm aux},\emptyset,\emptyset)$. Let $X$ be a solution for $\mathcal{I}^{\rm aux}$. Suppose there exists $v \in N^+(Q)\setminus N_Q$ that is reachable from $S$ in $D^{\rm aux}-X$, then observe that $X$ is also a solution for $\mathcal{I}^{\rm aux}_v$. Thus, {\sc Find-Branch-Set}$(\mathcal{I}^{\rm aux}_v)$ returns a branch set for~$\mathcal{I}$.
Now we look at the case when no vertex in $N^+(Q)\setminus N_Q$ is reachable from $S$ in $D^{\rm aux}-X$. Recall that $N_Q$ cannot be reachable from $S$ in $D^{\rm aux}-X$. Thus the entirety of $N^+(Q)$ is unreachable from $S$ in $D^{\rm aux}-X$. In this case, observe that $\mathcal{I}, Q,N^+(Q),X$ satisfy the conditions for Lemma~\ref{lem:pushingLemmaOne}, and thus, $N^+(Q)$ is a branch set for $\mathcal{I}$.

{\noindent \em Branching width: } It is at most $\vert N^+(Q) \vert =\lambda(\mathcal{I}^{\rm aux}) \leq k$.
%\todo[inline]{Prove measure drop here}

{\noindent \em Measure drop: } For each $v \in N^+(Q) \setminus N_Q$, we show that $\mu(\mathcal{I}^{\rm aux}_v) < \mu(\mathcal{I}^{\rm aux})$.
From Lemma~\ref{lem:closest_sep_lambda_inc}, the $\lambda(\mathcal{I}^{\rm aux}_v) \geq \lambda(\mathcal{I}^{\rm aux}) +1 $. However the size of $N_Q$ in the new instance decreases to $0$. Thus, the drop in $\mu$ is at least $(k^2+k-\lambda(\mathcal{I}^{\rm aux}) ^2 - \vert N_Q \vert) - (k^2+k-(\lambda(\mathcal{I}^{\rm aux}) -1)^2 - 0) = \lambda(\mathcal{I}^{\rm aux}) ^2 - \vert N_Q \vert - (\lambda(\mathcal{I}^{\rm aux}) +1)^2 = 2 \lambda(\mathcal{I}^{\rm aux})  - \vert N_Q \vert$. Since $\vert N_Q \vert \le \lambda $, the drop is $\geq \lambda +1$.

\medskip
{\noindent \bf Case 2.2: [$N^+(Q)$ is not the furthest minimum $S$-$T^{\rm aux}$ separator in $D^{\rm aux}$]: }From Lemma~\ref{lem:tightSuccessorComputation}, the algorithm first finds a separator $C'$ that tightly covers $N^+(Q)$. Let $Q' = R_D(S,C')$ that is $C'=N(Q')$. It then checks if there exists $F \subseteq D$ such that $F$ is isomorphic to some graph in $\HH$ and $r(F) \in Q'$. 

\smallskip
{\noindent \bf Case 2.2.1: [$\nexists$ a subgraph isomorphic to a graph in $\HH$ whose root is in $Q'$]: }In this case, the algorithm returns {\sc Find-Branch-Set}$(D^{\rm aux},\S,W,k,S,T^{\rm aux},Q', N_Q)$.

{\noindent \em Correctness: } Let $\mathcal{I}'^{\rm aux}= (D^{\rm aux},\S,W,k,S,T^{\rm aux},Q', N_Q)$. From construction, $N(Q')$ is a minimum $S-T^{\rm aux}$ separator and $Q'$ is root-$\cal H$-free. Also since $N^+(Q')$ covers $N^+(Q)$ and is an $S$-$T^{\rm aux}$ separator and for each $v \in N_Q$, $(v,t^*)$ is an arc in $D^{\rm aux}$, it follows that $N_Q \subseteq N^+(Q')$. Thus, if $X$ is a solution to $\mathcal{I}^{\rm aux}$, then it is also a solution to $\mathcal{I}'^{\rm aux}$.

{\noindent \em Branching width: } It is $1$.

{\noindent \em Measure drop: } Since the branching width is $1$ and $Q' \supset Q$, it is enough to prove that the measure does not increase. This is indeed the case, as $k, N_Q$ remain the same in both the instances. Also, since $S$ and $T^{\rm aux}$ remain the same, $\lambda(\mathcal{I}^{\rm aux}) = \lambda(\mathcal{I}'^{\rm aux})$.

\smallskip
{\noindent \bf Case 2.2.2: [$\exists F \subseteq D$ that is isomorphic to a graph in $\HH$ and $r(F) \in Q'$]: }
In this case, the algorithm returns $V(F)$ $\cup$ {\sc Find-Branch-Set}$(D^{\rm aux},\S,W,k,S \cup (N^+(Q)\setminus N_Q),T^{\rm aux} \cup \{r(F)\},\emptyset,\emptyset)$ $\cup$ $\bigcup_{v \in N^+(Q)\setminus N_Q}${\sc Find-Branch-Set}$(D^{\rm aux},\S,W,k,S,T^{\rm aux},Q,N_Q \cup \{v\})$ $\cup$ $\bigcup_{v \in V(F)}$ {\sc Pushing-Routine-2}$(D,\S,W,k,u)$. 

{\noindent \em Correctness: } Let $X$ be a solution of~$\mathcal{I}^{\rm aux}$. Since $X$ is a nice solution of $\mathcal{I}$, either $X \cap V(F) \neq \emptyset$ or, $X$ contains an $S_1$-$\{r(F)\}$ separator or, $X$ contains a $\{v\}$-$S_1$ separator, for some $v \in V(F)$.
In the first case, since $V(F)$ is present in the set returned, we are done. 
In the second case, since $S$ is reachable from $S_1$ in $D^{\rm aux}-X$, it follows that $X$ is an $S$-$\{T \cup r(F)\}$ separator. If there exists a vertex $v \in N^+(Q)\setminus N_Q$ that is unreachable from $S$ in $D^{\rm aux}-X$, then observe that $X$ is also a solution for {\sc Find-Branch-Set}$(D^{\rm aux},\S,W,k,S,T^{\rm aux},Q,N_Q \cup \{v\})$. Thus, we are done. Otherwise, $N^+(Q)\setminus N_Q$ is reachable from $S$ in $D^{\rm aux}-X$. In this case, $X$ is also a solution for {\sc Find-Branch-Set}$(D^{\rm aux},\S,W,k,S \cup (N^+(Q)\setminus N_Q),T^{\rm aux} \cup \{r(F)\},\emptyset,\emptyset)$, and hence, we are done.
In the third case, $X$ contains a $\{v\}$-$S_1$ separator, for some $v \in V(F)$. Using a similar argument as in Case 1, it follows that $u$ is in the forward-shadow of $X$ with respect to $W$ and that there is a $W$-$u$ path in $D$. Thus, from Lemma~\ref{lem:pushingRoutineTwo}, {\sc Pushing-Routine-2}$(D,\S,W,k,u)$ returns a branch set for $\mathcal{I}$.

{\noindent \em Branching width: } It is at most $ \vert N^+(Q) \vert +1 \leq \lambda(\mathcal{I}^{\rm aux}) +1 \leq k+1$.

{\noindent \em Measure drop: } Consider the instance $\mathcal{I}'^{\rm aux} = (D^{\rm aux},\S,W,k,S \cup (N^+(Q)\setminus N_Q),T^{\rm aux} \cup \{r(F)\},\emptyset, \allowbreak \emptyset)$. We show that $\mu(\mathcal{I}'^{\rm aux}) < \mu(\mathcal{I}^{\rm aux})$. From Lemma~\ref{lem:tight_sep_lambda_inc}, the minimum $(S \cup (N^+(Q) \setminus N_Q))$-$(T \cup \{r(F),t^{\rm aux}\})$ separator is of size greater than $\lambda(\mathcal{I}^{\rm aux})$. Thus, $\lambda(\mathcal{I}'^{\rm aux})> \lambda(\mathcal{I}^{\rm aux})$. However the  $N_Q$ (the last variable in the instance) for $\mathcal{I}'^{\rm aux}$ is an empty set. Therefore $\mu $ drops by at least $(k^2+k-\lambda(\mathcal{I}^{\rm aux}) ^2-\vert N_Q \vert)-(k^2+k - (\lambda(\mathcal{I}^{\rm aux}) +1)^2-0) = \lambda(\mathcal{I}^{\rm aux})^2 - \vert N_Q \vert - (\lambda(\mathcal{I}^{\rm aux})+1)^2 = 2 \lambda(\mathcal{I}^{\rm aux}) +1 - \vert N_Q \vert$. Since $\vert N_Q \vert \le \lambda(\mathcal{I}^{\rm aux}) $, the drop is at least $\lambda(\mathcal{I}^{\rm aux}) +1$. 
For each $v \in V(F)$, consider the instance $\mathcal{I}^{\rm aux}_v=(D^{\rm aux},\S,W,k,S,T^{\rm aux},Q,N_Q \cup \{v\})$. We now show that $\mu(\mathcal{I}^{\rm aux}_v) < \mu(\mathcal{I}^{\rm aux})$. Since $\vert N_Q \cup \{v\} \vert = \vert N_Q \vert +1$ and $\lambda(\mathcal{I}^{\rm aux}_v)=\lambda(\mathcal{I}^{\rm aux})$, we conclude that the measure drops by one in this case.

\bigskip
This concludes the description of the recursive algorithm together with its correctness.
%This finishes the description and correctness of the algorithm. Since each step of the algorithm can be done in polynomial time, to prove the time bound it remains to bound the number of nodes in the recursion tree of the algorithm. This bound will also bound the size of the branch set that the algorithm outputs. 
To bound the number of nodes in the recursion tree, since the maximum branching width of the recursion tree is at most $k+1$ and the depth of the recursion tree is at most $\mu(\mathcal{I}^{\rm aux})=k^2+k - \lambda(\mathcal{I}^{\rm aux})^2 - \vert N_Q \vert \leq k^2 +k$, we conclude that the number of nodes in the recursion tree is $(k+1)^{k^2 +k} = 2^{O(k^2 \log k)}$. Since the size of the set returned at the leaf nodes of the recursion tree is at most $k$ and at each level of the recursion tree a set of size at most $h+kh+1$ is added to the set obtained from the recursive calls, we conclude that the size of the set that the algorithm outputs is at most $(h+kh+1) \cdot 2^{O(k^2 \log k)} = O(h) \cdot 2^{O(k^2 \log k)}$.
\end{proof}
Given an instance $(D,\S,W,k)$ of {\sc r-$\HH$-SCCPC}, the algorithm of Lemma~\ref{lem:disjPartitionedCompressionSolution} creates an extended instance $\mathcal{I}^{\rm ext}=(D,\S,W,k, S_1, W \setminus S_1,\emptyset, \emptyset)$ of {\sc r-$\HH$-SCCPC} and calls {\sc Find-Branch-Set} on $\mathcal{I}^{\rm aux}$, which is an auxiliary instance of $\mathcal{I}^{\rm ext}$. The proof of Lemma~\ref{lem:disjPartitionedCompressionSolution} then follows from Lemma~\ref{lem:auxbranchset}.
As already argued, Lemma~\ref{lem:disjPartitionedCompressionSolution} implies Theorem~\ref{thm:rootedSCCDeletion}. This completes the description and proof of our {\sc FPT} algorithm for {\sc r-$\cH$-SCC Deletion}. Observe that the only place where we incur a $n^{O(h)}$ factor in the running time of this algorithm is for checking whether there exists a subgraph of $D$ isomorphic to a graph in $\cal H$. If this can be done in time $g(h)\cdot n^{O(1)}$, then our algorithm will run in time $g(k,h)\cdot n^{O(1)}$. In particular, if $\cal H$ comprises of only the $(d+1)$-out-star, the algorithm will run in time that is {\sc FPT} in $k$ and $d$ (Theorem~\ref{thm:outDegreeSCCDeletion}).

\subsection{Proof of Lemma~\ref{lem:disjPartitionedCompressionSolution}, Theorems~\ref{thm:rootedSCCDeletion} and~\ref{thm:outDegreeSCCDeletion}}

Given an instance $(D,\S,W,k)$ of {\sc r-$\HH$-SCCPC}, the algorithm of Lemma~\ref{lem:disjPartitionedCompressionSolution} creates an extended instance $\mathcal{I}^{\rm ext}=(D,\S,W,k, S_1, W \setminus S_1,\emptyset, \emptyset)$ of {\sc r-$\HH$-SCCPC} and calls {\sc Find-Branch-Set} on $\mathcal{I}^{\rm aux}$, which is an auxiliary instance of $\mathcal{I}^{\rm ext}$. The proof of Lemma~\ref{lem:disjPartitionedCompressionSolution} then follows from Lemma~\ref{lem:auxbranchset}.
As already argued, Lemma~\ref{lem:disjPartitionedCompressionSolution} implies Theorem~\ref{thm:rootedSCCDeletion}. This completes the description and the proof of our {\sc FPT} algorithm for {\sc r-$\cH$-SCC Deletion}. Observe that the only place where we incur a $n^{O(h)}$ factor in the running time of this algorithm is for checking whether there exists a subgraph of $D$ isomorphic to a graph in $\cal H$. If this can be done in time $g(h)\cdot n^{O(1)}$, then our algorithm will run in time $g(k,h)\cdot n^{O(1)}$. In particular, if $\cal H$ comprises of only the $(d+1)$-out-star, the algorithm will run in time that is {\sc FPT} in $k$ and $d$ (Theorem~\ref{thm:outDegreeSCCDeletion}).

\section{The FPT algorithm for {\sc Path $\cH$-SCC Deletion}}

\newcommand{\strc}{s.c.}
\newcommand{\yes}{{\sc Yes}}
\newcommand{\no}{{\sc No}}
\newcommand{\OO}{O}
In this section, we consider the $\HH$-SCC deletion problem when $\HH$ contains a directed path (the {\sc Path $\HH$-SCC} problem) and prove Theorem~\ref{thm:pathSCCDeletion}.

\pathSCCDeletion*

The overall picture of the arguments used is summarized here. Observe that what makes the $\HH$-SCC problem tricky is the fact that in the problem the goal is to delete a set of vertices thereby excluding the graphs in $\HH$ as subgraphs in the {\em strongly connected components} of the resulting graph. On the other hand, hitting/excluding the graphs of $\HH$ as subgraphs in the {\em whole graph} is relatively easy, as one can find a subgraph of $D$ isomorphic to a digraph in $\HH$, and then branch on its vertices to will go to the solution. We call the later problem as the {\sc $\HH$-Hitting} problem.

In what follows, we begin by showing that for every family of digraphs $\HH$, there exists another family $\HH^*$, such that $\HH$-SCC problem is equivalent to the {\sc $\HH^*$-Hitting} problem. The observation seems to be a good news but comes with own set of challenges, viz. the family $\HH^*$ might not be a finite family, therefore the simple branching algorithm described above might simply fail. Despite this nature of $\HH^*$ in the general case, we show that when $\HH$ contains a directed path, then the family $\HH^*$, as described above, exhibits certain nice properties (which we elaborate later in the section), that can be exploited to yield the FPT algorithm for the {\sc Path $\HH$-SCC} problem.

\subsection{The Design of $\HH^*$}

Given $\HH$, we will now define the family $\HH^*$, such that $\HH$-SCC problem reduces to the {\sc $\HH^*$-Hitting} problem. Towards this, given a digraph $H$, we define a class of strongly connected supergraphs of $H$, which we call the class of {\em path completions of $H$}, denoted by $PC(H)$. Intuitively, this is the class of digraphs obtained by adding paths between the vertices of $H$ to make the resulting graph strongly connected.

%First we will see how $\HH^*$ depends on $\HH$. 

%\todo[inline]{Define the union operation on graphs in preliminaries. Define a path $P$ as a graph.}

\begin{definition}\label{def:pathcompletion}
Let $H$ be a digraph. Then the {\em path completions of $H$}, denoted by $PC(H)$, is a class of strongly connected supergraphs of $H$, defined as follows. A supergraph $H^* \in PC(H)$ if $H^* = H \cup P_1 \cup \ldots \cup P_{\ell}$, where each $P_i$ is a directed path with end-points in $V(H)$. Also, the vertex set of the paths $P_i$ are not necessarily disjoint. The collection $\{P_i \colon i \in [\ell]\}$ is called the {\em witnessing collection of paths} for $H^*$. Note that there could be more than one witnessing collection of paths for a digraph.
\end{definition}

Note that, for any digraph $H$, the family $PC(H)$ could be potentially infinite. Thus, finding the collection $PC(H)$ is hard, but, as we will see later, checking if a digraph $H^* \in PC(H)$ is fairly easy, in the sense that one can check this in time proportional to the size of $H^*$.
We now refine this class $PC(H)$ to get a class that serves our purpose and avoids ``redundancy''.

\begin{definition}\label{def:goodpathcompletion}
Let $H$ be a digraph. The {\em good path completions of $H$}, denoted by $GPC(H)$, is a subset of $PC(H)$, such that, for each $H^* \in GPC(H)$, $H^* = H \cup P_1 \cup \ldots \cup P_{\ell}$ such that the pair of end-points of the paths $P_i$ are distinct. Thus, $\ell \leq \vert V(H) \vert^2$.
%Also, if $P_i$ is a $u$ to $v$ path, then there is no path from $u$ to $v$ in $H$. 
For a family of digraphs $\HH$, let $GPC(\HH) = \cup_{H \in \HH} GPC(H)$.
\end{definition}

%\todo[inline]{Give an english explanation of the definition of GPC}

The key insight of having the above definition is that for any strongly connected supergraph of $H$, say $\widehat{H}$, there exists $H^* \in GPC(H)$, such that $H^*$ is a subgraph of $\widehat{H}$. Also, the number of paths required to be added to $H$ to make it $H^*$ is bounded as a function of the size of $H$. The former claim is formalized below.

\begin{lemma}\label{lem:superpc}
Let $H$ be a digraph and let $\widehat{H} \supseteq H$ such that $\widehat{H}$ is strongly connected. Then there exists $H^* \in GPC(H)$ such that $H^* \subseteq \widehat{H}$.
\end{lemma}
\begin{proof}
We will construct the digraph $H^*$ iteratively. Begin by initialising $H^*:=H$. If $H^*$ is not strongly connected, then identify a pair $(u,v) \in V(H) \times V(H)$ such that there is no path from $u$ to $v$ in $H^*$. Since $\widehat{H}$ is strongly connected, there exists a $u$ to $v$ path, say $P_{u,v}$, in $\widehat{H}$. Update $H^* :=H^* \cup P_{u,v}$. 

From the construction above, clearly $H^* = H \cup P_1 \cup \ldots \cup P_{\ell}$, where the end-points of the paths $P_i$ are in $V(H)$. %and for each $P_i$, if $P_i$ is a path from $u$ to $v$ then there is no $u$ to $v$ path in $H$. 
What remains to prove now is that $H^*$ is strongly connected. Let us split the vertex set of $V(H^*)$ into two parts: $V(H)$ containing the vertices of $H$, and $V(P)$ containing the vertices in the paths $P_i$ that are not in $V(H)$, that is $V(P) = V(H^*) \setminus V(H)$. From the construction of $H^*$, it is clear that, for any two vertices $u,v \in V(H)$, there is both a $u$ to $v$, and $v$ to $u$ path in $H^*$. Since the endpoints of the paths $P_i$ are in $V(H)$, for every vertex in $V(P)$, there is path to a vertex in $V(H)$ and from a vertex in $V(H)$. Thus, every vertex of $V(P)$ can reach all other vertices of $V(H^*)$. This proves that $H^*$ is strongly connected.
\end{proof}

The above lemma together with the definitions stated above helps us to build the relation between the $\HH$-SCC problem and a corresponding {\sc $\HH^*$-Hitting} problem. We show that the $\HH^*$ corresponding to the family $\HH$ is $GPC(\HH)$. This is formalized below.

% \todo[inline]{In prelims---define the notation $H \subseteq D$ to say $H$ is a subgraph of $D$.}

%\todo[inline]{define $\uplus$}

\begin{lemma}\label{lem:scc_forbidden}
%$(D,k)$ is a \yes{} instance of $\HH$-SCC if and only if $(D,k)$ is a \yes{} instance of {\sc $GPC(\HH)$-Hitting} problem. In fact, 
Given a digraph $D$ and an integer $k$, $X \subseteq V(D)$ is a solution to the instance $(D,k)$ of the problem $\HH$-SCC if and only if it is a solution to the instance $(D,k)$ of the {\sc $GCP(\HH)$-Hitting} problem.
%Let $D$ be a graph. $D$ is in $\CC$ iff there does not exist a subgraph $F$ of $D$ that is isomorphic to some graph in $\HH^+$.
\end{lemma}

\begin{proof}
For the forward direction, let $X$ be a solution to the instance $(D,k)$ of $\HH$-SCC. Suppose, for the sake of contradiction, that there exists digraphs $H,H^*$, such that $H \in \HH$ and $H^* \in GPC(H)$, and $H^* \subseteq D-X$. Since $H^*$ is strongly connected (from the definition of $GPC$), there exists a strongly connected component, say $C$, of $D-X$ such that $H^* \subseteq C$. Since $H \subseteq H^*$ (from the definition of $GPC$), we conclude that $H \subseteq C$, thereby contradicting that $X$ is a solution for the problem $\HH$-SCC.
	
%Suppose there exists a subgraph $F$ of $D$ that is isomorphic to some graph in $\HH^+$. This means that $F$ is a path-completion of some graph $H \in \HH$ and thus $F$ is a strongly connected subgraph. Moreover, $F$ contains a subgraph isomorphic to $H$. Thus all the vertices of $H$ are in the same strong component of $D$ and thus $D$ is not in $\CC$.
For the backward direction, suppose that $X$ is a solution to the instance $(D,k)$ of the {\sc $GCP(\HH)$-Hitting} problem. For the sake of contradiction, suppose that there exists some strongly connected component, say $C$, of $D-X$ and some $H \in \HH$ and $H \subseteq C$. From Lemma~\ref{lem:superpc}, there exists a subgraph $H^* \in GPC(H)$ such that $H \subseteq H^* \subseteq C$. 
Since $C \subseteq D-X$ and $H^* \subseteq C$, we conclude that $H^* \subseteq D-X$, thereby contradicting that
%Thus, $H^* \in GPC(H)$ is present in $D-S$, which contradicts the fact that 
$X$ is a solution for the {\sc $GPC(\HH)$-Hitting} problem.
%	Suppose that $D \not\in \CC$, there must exist a witness subgraph $H \in \HH$ such that all vertices of $H$ are in the same strong component. Thus for any $u,v \in V(H)$ there exists a $u$-$v$ path $P_{uv}$ in $D$ and a $v$-$u$ path $P_{vu}$ in $D$. Let the collection of paths be $X = \{P_{uv},P_{vu}  \mid u,v \in V(H)\}$. The graph $F = (\bigcup_{P \in X} P) \cup H$ is a subgraph of $D$ that is strongly connected and contains $H$, hence it is a path-completion of $H$. Also note that $X$ induces a good set of paths by definition. Thus $F$ is isomorphic to some graph in $\HH^+$, proving the claim.
\end{proof}

\subsection{Discovering the structure of $GPC(\HH)$}
%\todo[inline]{VERY OPTIONAL-- May be change the notation of $\HH^+_p$ and $\HH^-_p$ to $GCP^+_p(\HH)$ and $GCP^-_p(\HH)$ respectively.}

As discussed earlier, even though for hitting problems it is relatively easier to design FPT algorithms, we cannot exploit Lemma~\ref{lem:scc_forbidden} forthright as the family $GPC(\HH)$ could contain digraphs of very large size and hence could be potentially infinite. Next, we split the family $GPC(\HH)$ into two parts such that one is a finite collection of digraphs of bounded size, and hence, exploitable by means of a branching algorithm, while the other part, which could be a potentially infinite collection, has a very special exploitable structure. Towards formalizing the above intuition, let $P \in \HH$ be some directed path of length $p$. Such a path exists because we work with the {\sc Path $\HH$-SCC} problem. Let $GPC(\HH) = \HH^+_p \uplus \HH^-_p$, such that for all the digraphs in $\HH^-_p$ there exists some witnessing collection of paths where all the witnessing paths have length at most $p-1$. Then, $\HH^+_p = GPC(\HH) \setminus \HH^-_p$. 
Recall that $h = \max_{H \in \HH} \vert V(H) \vert$.

\subsubsection{The finite sub-collection: $\HH^-_p$}
Lemma~\ref{lem:pc_bounded} concludes that the family $\HH^-_p$ is finite.

\begin{lemma}\label{lem:pc_bounded}
If $H^* \in \HH^-_p$, then $\vert V(H^*) \vert \leq h+(p-1)h^2$.
\end{lemma}

\begin{proof}
Consider any $H^* \in \HH^-_p$. Since, $\HH^-_p \subseteq GPC(\HH)$, $H^* \in GPC(H)$, for some $H \in \HH$. Also, since $H^* \in \HH^-_p$, there exists a witnessing collection of paths, say $P_1, \ldots, P_{\ell}$ such that $H^*= H \cup P_1 \cup \ldots \cup P_{\ell}$ and the length of each $P_i$ is at most $p-1$. Also, from the definition of $GPC$, $\ell \leq \vert V(H) \vert^2 \leq h^2$. Thus, $\vert V(H^*) \vert \leq h+ (p-1)h^2$.
%Note that since each subgraph in $\FF_D$ is strongly connected, it must be the case that every subgraph in $\FF_D$ is contained in a single strong component of $D$. Thus by assumption, no subgraph in $\FF_D$ contains a path $P_d$ as a subgraph. We claim that every such subgraph has size atmost $|H|+d|H|^2$. Since every subgraph in $\FF_D$ is a path-completion $H'$ of some graph $H$ in $\HH$. Thus we can write $H'$ as $H \cup P_1 \cup \ldots \cup P_l$ where $P_1,...,P_l$ are paths with endpoints in $H$. Moreover, the size of each path $P_i$ is bounded by $d$. We have seen earlier that $l \le |H|^2$. Since each path adds atmost $d$ vertices to the graph, the total number of vertices possible in the path-completion are bounded by $|H| + d|H|^2 \le p+dp^2$. 
	% First we count all good path-completions of some graph $H \in \HH$ such that all paths in the path-completion are of size less than $d$. Let $P_1,\ldots,P_l$ be the paths in the path-completion.
\end{proof}
%the all the witnessing paths of the digraphs in $\HH^-_p$ have length at most $p$ and $\HH^+_p$ contains all the other digraphs of $GPC(\HH)$, that is the digraphs which have atleast one witnessing path of length strictly greater than $p$.

From Lemma~\ref{lem:pc_bounded} one can derive an easy algorithm for computing the family $\HH^-_p$.

\begin{lemma}\label{lem:computingHneg}
Given the family of digraphs $\HH$, the family $\HH^-_p$ can be computed in $2^{O(h^6)}$ time.
%$f(h)$ time for some function $f$ that only depends on $h$.
\end{lemma}
\begin{proof}
	From Lemma~\ref{lem:pc_bounded}, note that $\HH^-_p$ contains only digraphs of size at most $h + (p-1)h^2$ which are in $GPC(\HH)$. Thus, to enumerate the family $\HH^-_p$ it is enough to enumerate all digraphs of size at most $h + (p-1)h^2$, and then for each of them check whether it is in $GPC(\HH)$. 
	In order to check if a digraph $H^*$ is in $GPC(\HH)$, first check whether $H^*$ is strongly connected followed by guessing the partition of the digraph $H^*$ into at most $h^2 +1$ parts, say $H \cup P_1 \cup \ldots \cup P_{\ell}$, where $|V(H)| \le h$ and $\ell \leq |V(H)|^2$, and checking if $H \in \HH$ and the paths $P_i$ are of size at most $p-1$ and have distinct pair of end-points in $V(H)$. 

	Since the number of digraphs on at most $h+(p-1)h^2$ vertices is at most $2^{(h+(p-1)h^2)^2}$ and $p \leq h$, and each of the steps described above takes time at most $2^{O(h^6)}$, the running time follows.
	%proportional to the size of $H^*$, we conclude that the family $\HH^-_p$ can be enumerated in $f(h)$ time, for some function $f$ depending entirely on $h$.
\end{proof}

\subsubsection{Making the instance $\HH^-$-free}
Here, we design a branching algorithm that takes an instance $(D,k)$ of $\HH$-SCC and returns an equivalent instance $(D',k')$ of $\HH$-SCC such that $D'$ has no digraph in $\HH^-_p$ as a subgraph. We call such a digraph $\HH^-_p$-free. Also, $k' \leq k$.

\begin{lemma}\label{lem:hh-free}
	Let $(D,k)$ be an instance of $\HH$-SCC. In time $2^{O(h^6)} \cdot h^{O(k)} \cdot n^{O(h^3)}$, one can either conclude that $(D,k)$ is a no-instance of $\HH$-SCC or output at most $h^k$ instances $\{(D_1,k_1),\ldots,(D_q,k_q)\}$ of $\HH$-SCC such that for each $i \in [q]$, $D_i$ is $\HH^-_p$-free and $k_i \le k$, and $(D,k)$ is a yes-instance if and only if there exists $i \in [q]$ such that $(D_i,k_i)$ is a yes-instance.
\end{lemma}
\begin{proof}
Let $(D,k)$ be an instance of $\HH$-SCC. 
Compute the family $\HH^-_p$ using Lemma~\ref{lem:computingHneg}.
Suppose $D$ contains a subgraph, say $F$ isomorphic to some graph in $\HH^-_p$.
%, such that $F$ is isomorphic to some digraph $H^* \in \HH^-_p$.
Since $\HH^-_p \subseteq GPC(\HH)$, from Lemma~\ref{lem:scc_forbidden}, there exists a solution $X$ of the instance $(D,k)$ of $\HH$-SCC such that $X$ contains some vertex of $F$.
%a vertex of the subgraph of $D$ isomorphic to $H^*$. 
%Let $F$ be the subgraph isomorphic to $H^*$ inn $D$. 
	For each $v \in V(F)$, the algorithm branches in the following instances: $(D-\{v\},k-1)$. Since there exists a solution containing some vertex of $v$, $(D,k)$ is a yes-instance if and only if at least one of $(D - \{v\},k-1)$, for $v \in V(F)$ is a yes-instance. The branching algorithm stops when $k \leq 0$ or when the resulting digraph has no subgraph isomorphic to a digraph in $\HH^-_p$. When $k\leq 0$, if the resulting digraph has a subgraph isomorphic to a digraph in $\HH^-_p$, then report that $(D,k)$ is a no-instance of $\HH$-SCC. This completes the description of the algorithm. For the running time analysis, since the size of the digraphs in $\HH^-_p$ is at most $h+ph^2$, $p \leq h$ and there are at most $2^{(h+(p-1)h^2)^2}$ graphs in $\HH^-_p$s, we can check whether there exists a subgraph $F$ isomorphic to a graph in $\HH^-_p$ in time $2^{(h+(p-1)h^2)^2}\cdot n^{O(h^3)}$.	
	%a subgraph $F$ isomorphic to $H^*$ can be found in $n^{O(h^3)}$ time. 
%	Furthermore since there are at most $2^{(h+(p-1)h^2)^2}$ graphs in $\HH^-_p$, we can check whether there exists a subgraph $F$ isomorphic to a graph in $\HH^-_p$ in time $2^{(h+(p-1)h^2)^2}\cdot n^{O(h^3)}$.
 Since the branching algorithm stops when $k\leq 0$, we get the following recurrence: $T(k) \leq \sum_{i \in [h]} T(k-1)$, $T(0)=1$, where $T(k)$ denotes the number of leaves in the branching tree rooted at the instance with budget parameter $k$. By induction, one can prove that $T(k) \leq h^k$. This yields the desired running time.
\end{proof}
Henceforth, we assume that the instance $(D,k)$ of $\HH$-SCC is such that $D$ is $\HH^-_p$-free.

\subsubsection{The structure of the potentially infinite sub-collection: $\HH^+_p$}
Recall that $\HH^+_p$ is a collection of digraphs in $GPC(\HH)$ which have a witnessing collection of paths where at least one path has length strictly more than $p$.
\begin{lemma}\label{lem:pc_subgraph}
For each $H^* \in \HH_p^+$, there exists a subgraph $H' \subseteq H^*$ such that $H' \in GPC(P)$.
%, where $\widehat{P}$ is a directed path on at least $p$ vertices.
(Recall $P$ is a directed path in $\HH$ that we fixed.)
\end{lemma}
\begin{proof}
Since $H^* \in \HH^+_p$, let $H^*$ be a path-completion of $H \in \HH$ where one of the witnessing paths has length at least $p$. That is, $H^* =H \cup P_1 \cup \ldots \cup P_{\ell}$, where there exists $i \in [\ell]$ such that $\vert V(P_i) \vert \geq p$. Without loss of generality, let $i =1$. Since the length of $P$ is $p$, $P$ is a subpath of $P_1$. Let $P$ be a path from $u$ to $v$ in $H^*$.
%Let $u$ and $v$ be the endpoints of $P_1$ in $V(H)$. 
 Since $H^*$ is strongly connected, there exists another path say $P'$ from $v$ to $u$. Then $P^*= P \cup P'$ is a closed walk in $H^*$. Observe that $P^*$ is strongly connected. Also by construction, $P^* \in GPC(P)$ (with $P'$ being the witnessing path). Since $P^* \subseteq H^*$, we are done.
%We claim that $H'$ contains a subgraph that is a path-completion of $P_d$. The path $P_i$ connects some vertex $u \in H$ to some vertex $v \in H$. Due to strong connectivity of $H'$, there exists another path $P' \in H'$ from $v$ to $u$. Thus $X=P_1 \cup P_2$ forms a closed walk that contains a subpath of length $d$. This closed walk is the desired path-completion of $P_d$.
\end{proof}

Combining Lemmas~\ref{lem:pc_subgraph} and~\ref{lem:scc_forbidden} we get the following lemma.

\begin{lemma}\label{lem:enough_hit_pc}
Let $(D,k)$ be an instance of {\sc Path $\HH$-SCC} such that $P \in \HH$ is a directed path of length $p$ and $D$ is $\HH^-_p$-free. Then, $(D,k)$ is a yes-instance of {\sc Path $\HH$-SCC} if and only if it is a yes-instance of {\sc $GPC(P)$-Hitting}.
%, where $P$ is some path in $\HH$ of length $p$.
\end{lemma}
\begin{proof}
Recall that $GCP(\HH) = \HH^-_p \uplus \HH^+_p$. For the forward direction, let $X$ be a solution  to the instance $(D,k)$ of {\sc Path $\HH$-SCC}. Since $P \in \HH$, from Lemma~\ref{lem:scc_forbidden}, $X$ is also a solution to the instance $(D,k)$ of {\sc $GPC(P)$-Hitting}. For the backward direction, let $X$ be a solution to the instance $(D,k)$ of {\sc $GPC(P)$-Hitting}. We first prove that $X$ is also a solution of {\sc $\HH^+_p$-Hitting}. This follows from Lemma~\ref{lem:pc_subgraph}. Since $X$ is a solution of {\sc $\HH^+_p$-Hitting}, and $D$ is $\HH^-_p$-free, $X$ is also a solution of $\HH$-SCC from Lemma~\ref{lem:scc_forbidden}. This proves the lemma.
%Since $D$ is $\HH^-_p$-free, from Lemma~\ref{lem:scc_forbidden}, $(D,k)$ is a \yes{} instance of $\HH$-SCC if and only if it is a \yes{} instance of {\sc $\HH^+_p$-Hitting}. 
%From Lemma~\ref{lem:pc_subgraph}, every set $S\subseteq V(D)$ that hits the digraphs in $\HH^+_p$ also hits the digraphs in $GPC(P)$. Also, since $P \in \HH$, from Lemma~\ref{lem:scc_forbidden}, any solution of the $\HH$-SCC problem, hits all digraphs in $GPC(P)$. This proves the lemma.
\end{proof}

Combining Lemmas~\ref{lem:enough_hit_pc} and~\ref{lem:scc_forbidden}, we get the following lemma.

\begin{lemma}\label{lem:finalP}
Let $D$ be a $\HH^-_p$-free digraph. Then, $(D,k)$ is a yes-instance of $\HH$-SCC if and only if it is a yes-instance of $\{P\}$-SCC.
\end{lemma}
\begin{proof}
From Lemma~\ref{lem:finalP}, $(D,k)$ is a yes-instance of $\HH$-SCC if and only if it is a yes-instance of {\sc $GPC(P)$-Hitting}. Also, from Lemma~\ref{lem:scc_forbidden}, $(D,k)$ is a yes-instance of {\sc $GPC(P)$-Hitting} if and only if it is a yes-instance of $\{P\}$-SCC. Combining both the statements above, we prove the lemma.
\end{proof}

% \todo[inline]{Give the correct label below for the theorem of path- h-scc, rooted-h-scc from intro. also the precise running time of rooted from intro}

\subsection{Proof of Theorem~\ref{thm:pathSCCDeletion}}

Let $(D,k)$ be an instance of {\sc Path $\HH$-SCC} and let $P \in \HH$ be a directed path of length $p$. Let $h = \max_{H \in \HH} \vert V(H) \vert$.
From Lemma~\ref{lem:hh-free}, in time $2^{O(h^6)} \cdot h^{O(k)} \cdot n^{O(h^3)}$, we get a set of instances $\{(D_1,k_1),\ldots,(D_q,k_q)\} $, such that for each $i \in [q]$, $D_i$ is $\HH^-_p$-free and $k_i \leq k$, and $(D,k)$ is a yes-instance if and only if for some $i \in [q]$, $(D_i,k_i)$ is a yes-instance. From Lemma~\ref{lem:finalP}, we conclude that it is enough to solve the $\{P\}$-SCC problem on these instances to obtain the solution. Since $P$ is a rooted digraph, from Theorem~\ref{thm:rootedSCCDeletion}, the problem can further be solved in $2^{O(k^{3}\log k)} \cdot n^{O(p)}$. Thus, we get an algorithm with running time $2^{O(k^{3}\log k)} \cdot h^{O(k)} \cdot 2^{O(h^6)} \cdot n^{O(h^3)}$.

\section{Faster FPT algorithms}
\subsection{Faster FPT algorithm for {\sc 1-Out-regular Deletion}}\label{sec:fasterAlgorithmOutRegularDeletion}

In this section, we give an algorithm for {\sc Rooted $\cH$-SCC Deletion}. 
%Theorem~\ref{thm:rootedSCCDeletion} 
that runs in time $O^*(2^{O(k \log k)})$ when $\cH$ contains only the out-directed {2-star}, that is we prove Theorem~\ref{thm:oneOutRegularDeletion}.

\oneOutRegularDeletion*

In the following (Definition~\ref{def:outRegularniceInstances}, Observation~\ref{obs:uniqueReachableVertex}, Lemma~\ref{lem:disjointnessLemma}-~\ref{lem:mainOneOutRegularDeletionAlgorithm}), fix $\tau=(D,\S=(S_1,\dots$, $S_q),W,k)$ to be an instance of {\sc r-$\cH$-SCC PCC}. 
 Recall that a solution for $\tau$ is a set $X\subseteq V(D)\setminus W$ of size at most $k$ that  intersects all $S_i$-$S_j$ paths in $D$ for every $j>i$ such that every non-trivial strongly connected component of $D-X$ is $\cH$-free and hence, is a cycle.

We have the following specialization of Definition~\ref{def:niceInstances} to our current choice of $\cH$. 

\begin{definition}
\label{def:outRegularniceInstances}

	A solution $X$ for the instance $\tau$ is said to be {\em nice} if for every  triple $u,v,w\in V(D)$  such that $v,w\in N^+(u)$ and for every $i\in [q]$, one of the following holds.
	
	\begin{enumerate}
	\item $X$ intersects $\{u,v,w\}$.
	\item $u\notin R(S_i,X)$.
	\item  There is no $v$-$S_i$ path in $D-X$ or no $w$-$S_i$ path in $D-X$.
	\end{enumerate}

\end{definition}

Using Lemma~\ref{lem:DisjCompressionToDisjPartionedCompression}, Theorem~\ref{thm:oneOutRegularDeletion} can be obtained as a consequence of the following lemma.

\begin{lemma}\label{lem:fullMainOneOutRegularDeletionAlgorithm}
There is an algorithm that, given $\tau$, runs in time $2^{O(k\log k)} \cdot n^{O(1)}$ and either returns a solution for $(D,k)$ or correctly concludes that there is no nice solution for $\tau$.
\end{lemma}

The rest of Section~\ref{sec:fasterAlgorithmOutRegularDeletion} is therefore devoted to proving Lemma~\ref{lem:fullMainOneOutRegularDeletionAlgorithm}.
We assume without loss of generality that for every $v \in V(D)$,  either $v$ lies in a strongly connected component of $D$ that intersects $S_1$ and contains at least one vertex of out-degree at least 2 or $v$ can reach $W\setminus S_1$ in $D$.
This can be ensured by a straightforward preprocessing routine that computes the strongly connected components of $D$ and deletes vertices that do not satisfy these properties. The correctness of this step follows from the fact that the deleted vertices do not participate in minimal solutions for the given instance and moreover, the non-existence of a nice solution in the reduced instance is not affected by adding back the deleted vertices. 
We begin with the following simple observation regarding graphs with maximum out-degree 1. 

\begin{observation}\label{obs:uniqueReachableVertex}
Let $R\subseteq V(D)$ be such that every vertex in $R$ has out-degree at most 1 in $D$ and let $Z\subseteq N^+[R]$. Then, each vertex of $R$ can reach at most one vertex of $Z$ via paths whose internal vertices (if there are any) are contained in $R\setminus Z$.
\end{observation}

\begin{proof}
If this were not true, then there would be a vertex in $R$ with at least two out-neighbours in $D$, a contradiction to the premise.
\end{proof}

In particular, the above observation implies that each vertex of $R$ can reach at most one vertex of $N^+(R)$ via paths whose internal vertices are contained in $R$. 
Motivated by  Observation~\ref{obs:uniqueReachableVertex}, we define the following notation.

\begin{definition}\label{def:closeVertices}
	 Let $S_1\subseteq R\subseteq V(D)$ be such that every vertex in $R$ (i) is reachable from $S_1$ in $D[R]$ and (ii) has out-degree at most 1 in $D$. For every $s\in S_1$, we define  ${\emph{close}(s,R)}=(R\setminus S_1)\cap N^-(s)$. The notation is extended to subsets of $S_1$ in a natural way. 
\end{definition}

As a consequence of Observation~\ref{obs:uniqueReachableVertex}, we have that each $s\in S_1$ can reach at most one vertex of $close(S_1,R)$ via a path internally vertex-disjoint from $close(S_1,R)$ (and hence also disjoint from $S_1$). Therefore, $|close(S_1,R)|\leq |S_1|$.	

%\begin{proof}
%Suppose to the contrary that there is a vertex $s\in R$ that cannot reach $N^+(R)$. Recall that either $s$ lies in a strongly connected component of $D$ that intersects $S_1$ and contains at least one vertex of out-degree at least 2 or $s$ can reach $W\setminus S_1$ in $D$.
%% via paths whose internal vertices are contained in $R$. 
%%Since every vertex in $D$ lies in a strongly connected component of $D$ that contains at least one vertex of out-degree at least 2, it must be the case that $s$ lies on a closed walk in $D$ that contains a  vertex of out-degree at least 2. Since every vertex in $R$ has  out-degree at most 1 in $D$, 
%
%In the former case,   $s$ lies on a closed walk in $D$ that contains a  vertex of out-degree at least 2. Therefore, in either case, we 
%infer that $s$ can reach 
%at least one vertex in $V(D)\setminus R$, implying that $s$ can indeed reach $N^+(R)$. Observation~\ref{obs:uniqueReachableVertex} already guarantees that at most one vertex of $N^+(R)$ can be reachable from $s$ by such paths. This completes the proof of the observation.\end{proof}

 Fix a set  $R$ satisfying the premise of this observation. For every $v\in R\setminus S_1$ that can reach a vertex $s\in S_1$ (which must then be unique) via a path contained in $R$ and internally vertex-disjoint from $S_1$, we denote by $\partial(v)$ the singleton set containing the unique vertex of $close(s,R)$ that lies on this $v$-$s$ path in $D[R]$. For every $v\in R\setminus S_1$ that can reach a vertex of $N^+(R)$ via a path internally vertex-disjoint from  $N^+(R)\cup S_1$, we denote by $\partial(v)$ the singleton set containing this unique vertex of $N^+(R)$. For every other $v\in R\setminus S_1$, it must be the case that $v$ cannot reach reach a vertex of $N^+(R)$ via a path internally vertex-disjoint from  $N^+(R)$ and so, we set 
$\partial(v)=\emptyset$.

In other words, for every $v\in R\setminus S_1$, we define $\partial(v)$ as follows. 
We consider the unique maximal path contained in $D[R]$ that starts at $v$ and is disjoint from $S_1$. Suppose that this path terminates at the vertex $w$. If $N^+(w)$ is empty or only comprises $v$, then $\partial(v)$ is defined as $\emptyset$. Otherwise, $\partial(v)=N^+(w)$. 

%$w\in N^-(S_1)$, then 
%
%$\in N^-(S_1\cup N^+(R))$. Then, $\partial(v)$ is defined as $N^+(w)$. Otherwise, $\partial(v)=\emptyset$. 
%%\todo[inline]{Shoudnt it be $N^-(close(S_1,R) \cup N^+(R))$?}
%%\red{Say something about the case that $v$ lies in a sink scc that's a cycle.}
%%
%
%
%
%
%, we denote by $\partial(v)$ the unique vertex of $N^+(R)$ reachable from $v$ via paths with internal vertices disjoint from $N^+(R)$.  Similarly, for every $v\in R$, that can reach a vertex of $S_1$ via a path disjoint from $N^+(R)$, we denote by $\chi(v)$ this unique vertex of $S_1$. The uniqueness follows from Observation~\ref{obs:uniqueReachableVertexConsequence}. For the remaining vertices in $R$, we define $\chi(v)=\partial(v)$.
 For a set $R'\subseteq R$, we denote by $\partial(R')$ 
% and 
% $\chi(R')$ 
 the set $\bigcup_{v\in R'}\partial(v)$.
%  and $\bigcup_{v\in R'}\chi(v)$ respectively. 
{Notice} that for every $R'\subseteq R$,  $\partial(R')\cap R\subseteq close(S_1,R)$  and $|\partial(R')|\leq |R'|$.

The following lemma forms the crux of the correctness of our main algorithm. The lemma identifies a pair of vertex subsets in the graph such that if there is a certain kind of nice solution for $\tau$, then, we may assume that for either of these sets, the intersection of the nice solution with the set is one of only a bounded (in $k$) number of possibilities.

\begin{lemma}\label{lem:disjointnessLemma}
%Let $\tau=(D,\S=(S_1,\dots, S_q),W,k)$ be an instance of {\sc 1-Out-Regular SCC Deletion}. 
Let $T\supseteq W\setminus S_1$, $\L\subseteq V(D)\setminus T\cup \{S_1\}$  and let $C$ be a minimal $S_1$-$T$ separator in $D$ that is disjoint from $\L$. Let $R=R(S_1,C)$. 
	Suppose that every vertex in $R$ has out-degree at most 1 in $D$ and suppose that there is a nice solution $X$ for $\tau$ that is an $S_1$-$T$ separator and an $\L$-$T\cup \{S_1\}$ separator. Then the following statements hold:
	\begin{enumerate}\item There is a nice solution $X'$ for $\tau$ that is an $S_1$-$T$ separator, an $\L$-$T\cup \{S_1\}$ separator and moreover, $X'\cap R\subseteq$ $close(S_1,R)$. 
	\item 	There is a nice solution $X'$ for $\tau$ that is an $S_1$-$T$ separator, an $\L$-$T\cup \{S_1\}$ separator and moreover, $X'\cap R_{max}(\L,T\cup \{S_1\})=\emptyset$. 
	\item If $c\in V(D)$ is  reachable from $S_1$ and can reach $S_1$ in $D-X$, then $X$ is also a nice solution for the instance $\tau'=(D,\S=(S_1\cup \{c\},\dots$, $S_q),W,k)$ that is an $S_1\cup \{c\}$-$T$ separator and an $\L$-$T\cup \{S_1\}\cup \{c\}$ separator. 
	\end{enumerate}
%\red{Address the possibility that $\L\cap C\neq \emptyset$.}
\end{lemma}

\begin{proof}
We begin by observing that $R$ satisfies the properties in the premise of Observation~\ref{obs:uniqueReachableVertex} and Definition~\ref{def:closeVertices}.

Consider the first statement. Suppose that $X_C=X\cap R$ is non-empty. We claim that $X'=(X\setminus X_C)\cup \partial(X_C)$ is also a nice solution for this instance that is an $S_1$-$T$ separator and an $\L$-$T\cup \{S_1\}$ separator. Since $\partial(X_C)$ is no larger than $X_C$, it follows that     $|X'|\leq |X|$. 

	We now argue that $X'$ is a nice solution, an $S_1$-$T$ separator and an $\L$-$T\cup \{S_1\}$ separator. If not, then it must be the case that either (i) there is a closed walk $\cW$ in $D-X'$ that 
	induces a subgraph containing a vertex of out-degree at least 2  
%	\todo{out-degree atleast 2 and both neighbours are in the same SCC?}  
	or (ii) there is an $S_i$-$S_j$ path $\cW$ in $D-X'$, where $i<j$ (implying that $X'$ is not a solution for $\tau$) or (iii) there is an $S_1$-$T$ path $\cW$ in $D-X'$ or (iv) there is a $\L$-$T\cup \{S_1\}$ path $\cW$ in $D-X'$ or (v) there is an $i\in [q]$ and a triple $u,v,w\in V(D)$ with $v,w\in N^+(u)$ such that $X'$ is disjoint from $\{u,v,w\}$ and $D-X'$ contains  an $S_i$-$u$ path $\cW_u$, a $v$-$S_i$ path $\cW_v$ and a $w$-$S_i$ path $\cW_w$, implying that $X'$ is not a {\em nice} solution for $\tau$
 (see Definition~\ref{def:outRegularniceInstances}). 
% Notice that by the premise of the lemma, in  Case (v), $\{u,v,w\}\setminus R\neq \emptyset$. 

%Here, (i) and (ii) imply that $X'$ is not a nice solution, (iii) implies that $X'$ is not an  $S_1$-$T$ separator and  (iv) implies that $X'$ is not an $\L$-$W$ separator.
	Observe that in each of the cases (i) to (iii), $\cW$ must contain some $x\in X_C$ and some  $y\in V(D)\setminus R$ such that $\cW$ contains an $x$-$y$ subwalk.
%	\todo{$C$ can be of degree 2, so shouldnt it be $V(D)\setminus R$?}.
	 Therefore, we conclude that  in each of these cases, $\cW$ must intersect $\partial(x)$ for some $x\in X_C$, a contradiction since $\partial(x)\in X'$. 
 
 We now consider Case (iv). In this case as well, $\cW$ must intersect $X_C$ at a vertex $x$. Moreover, if $\cW$ contains 
	an $x$-$y$ walk for some $y\in V(D)\setminus R$ 
%	\todo{$V(D)\setminus R$ should be enough}
	, then the same argument as above implies a contradiction.  Hence, we may assume that $\cW$ is an $\L$-$S_1$ path and moreover, the subpath of $\cW$ from $x$ to $S_1$ is contained in $R$. However, the unique vertex of $N^{-}(S_1)$ that comprises the set $\partial(x)$, must also be contained in this subpath, a contradiction since $\partial(x)\subseteq  X'$.

	We now consider Case (v). By the premise of the lemma and the fact that $u$ has out-degree at least 2, we have that $u\notin R$. If $u$ is not reachable from $S_i$ in $D-X$, then we have a contradiction along the same lines as that used before. That is, we have a path $\cW$ in $D-X'$ from some $x\in X_C$ to some $y\in V(D)\setminus R$.
%	 \todo{$u$ can be in $C$, why not $V(D)\setminus R$?}
 Moreover, the same argument also implies that $i=1$ (since $R\subseteq V(D)\setminus T$ and $T\supseteq W\setminus S_1$). That is, $D-X'$ contains  an $S_1$-$u$ path $\cW_u$, a $v$-$S_1$ path $\cW_v$ and a $w$-$S_1$ path $\cW_w$. We have already concluded that $D-X$ also contains an $S_1$-$u$ path. Without loss of generality, suppose that $D-X$ does not contain the path $\cW_v$. Let $x$ be the last vertex of $X_C$ that lies on this path when traversing it from $v$ to $S_1$. Then, the subpath of $\cW_v$ from $x$ to $S_1$  also intersects $\partial(x)$, a contradiction since $\partial(x)\subseteq X'$. This completes the proof of the first statement.

We now prove the second statement. Let $\hat R=R_{max}(\L,T\cup \{S_1\})$ and $\hat C=C_{max}(\L,T\cup \{S_1\})$.
Suppose that $X_{\hat C}=X\cap \hat R$ is non-empty. Furthermore, let $J=\hat C\setminus R(\L,X)$. That is, $J$ comprises those vertices of $\hat C$ that are {\em not} reachable from $\L$ after deleting $X$. Since $\hat C$ is a minimum $\L$-$T\cup \{S_1\}$ separator, it must be the case that $|X_{\hat C}|\geq |J|$.  This bound on $|J|$ can be immediately inferred from the existence of a set of $|\hat C|$ pairwise internally vertex disjoint $\L$-$T\cup \{S_1\}$ paths.

 We claim that $X'=(X\setminus X_{\hat C})\cup J$ is also a nice solution for this instance that is an $S_1$-$T$ separator and an $\L$-$T\cup \{S_1\}$ separator. Since $|X_{\hat C}|\geq |J|$, it follows that     $|X'|\leq |X|$. 

	It remains to be  argued that $X'$ is a nice solution, an $S_1$-$T$ separator and an $\L$-$T\cup \{S_1\}$ separator. If this were not the case, then one of the five cases (i)--(v) enumerated earlier (see proof of the first statement of this lemma) must hold in $D-X'$. In each of these cases, we infer the presence of an $x$-$y$ walk in $D-X'$ where $x\in X_{\hat C}$ and $y\in T\cup \{S_1\}$. But this implies the presence of an $x'$-$y$ walk in $D-X$ where $x'\in \hat C\setminus J$. However, by the definition of $J$, we have that $\hat C\setminus J\subseteq R(\L,X)$. This is a contradiction to $X$ being an $\L$-$T\cup \{S_1\}$ separator. This completes the argument for the second statement.

We now consider the third statement. If $X$ is not a solution for the tuple $\tau'$, then it must be the case that there is a $c$-$d$ path in $D-X$ for some $d\in \bigcup_{i\in [q]\setminus \{1\}}S_i$. Since $c$ is reachable from $S_1$ in $D-X$, this implies an $S_1$-$d$ path, a contradiction to $X$ being a solution for $\tau$. On the other hand, suppose that $X$ is not a {\em nice} solution for $\tau'$. Then, there is an $i\in [q]$ and a triple $u,v,w\in V(D)$ with $v,w\in N^+(u)$ such that $X$ is disjoint from $\{u,v,w\}$ and $D-X$ contains  an $S_1\cup \{c\}$-$u$ path $\cW_u$, a $v$-$S_1\cup \{c\}$ path $\cW_v$ and a $w$-$S_1\cup \{c\}$ path $\cW_w$. Since $c$ is both reachable from $S_1$ and can reach $S_1$ in $D-X$, we conclude that $D-X$ contains an $S_1$-$u$ path, a $v$-$S_1$ path and a $w$-$S_1$ path. This contradicts the premise that $X$ is a nice solution for $\tau$.  This completes the proof of the lemma.
\end{proof}

We now provide a subroutine that computes a set of vertices upon which our main algorithm will be able to branch exhaustively while strictly making progress in each branch.

\begin{lemma}\label{lem:branchableObjectOutDegreeDeletion}Let $T\supseteq W\setminus S_1$,
There is an algorithm that, given $\tau$ and $T$,  runs in polynomial time and performs one of the following operations: 
\begin{enumerate}\item Outputs a $u_1\in R_{min}(S_1,T)$ and  $u_2,u_3\in N^+_D(u)$. 
\item 	Correctly concludes that $R_{max}(S_1,T)$ has no vertex with out-degree at least 2 in $D$. 
\item Outputs a minimum $S_1$-$T$ separator $C$ such that $R(S_1,C)$ has no vertex with out-degree at least 2 in $D$ and moreover,  for some $u_1\in C$, there exist  $u_2,u_3\in N^+_D(u_1)\setminus T$.
\end{enumerate}
\end{lemma}

\begin{proof}
We first compute $R_{min}(S_1,T)$ and $R_{max}(S_1,T)$ and  check whether there exist  (i) $u_1\in R_{min}(S_1,T)$,   $u_2,u_3\in N^+_D(u)$ or (ii) $u\in R_{max}(S_1,T)$ such that $u$ has out-degree at least 2 in $D$.  If the answer to (i) is affirmative, then we return $u_1,u_2, u_3$. Similarly, if the answer to (ii) is negative, then we  return the same.

%Suppose that neither holds. Then, one can 
We now greedily compute a minimum $S_1$-$T$ separator $C$ such that $R(S_1,C)$ has no vertex with out-degree at least 2 in $D$ and $R(S_1,C)$ is maximal subject to the out-degree constraint on the vertices contained within. 
%We now argue that for every $u_1\in C$, there exist  $u_2,u_3\in N^+_D(u_1)$.
 Recall that we are in the case where 
 $C\neq R_{max}(S_1,T)$. Therefore, $C_{max}(S_1,T)$ covers $C$. But this implies that  there is some $u_1\in C\cap {R_{max}(S_1,T)}$. Notice that $u_1$ has no out-neighbours in $T$. 

Therefore, it suffices to prove that $u_1$ has out-degree at least 2 in $D$. Suppose to the contrary that  $u_1$  has at most one out-neighbour in $D$. Since $C$ is a minimal $S_1$-$T$ separator, it follows that $u_1$ has at least one out-neighbour in $D$.
 Therefore, we have that $u_1$ has exactly one out-neighbour in $D$, call it $u_2$. We now observe that $C'=(C\setminus \{u_1\})\cup \{u_2\}$ is also a minimum $S_1$-$T$ separator, $R(S_1,C')\supset R(S_1,C)$ and every vertex in $R(S_1,C')$ has out-degree at most 1 in $D$. This contradicts our choice of $C$,  completing the proof of the lemma.
\end{proof}

We are now ready to present the main algorithm of this section. Specifically, we obtain Lemma~\ref{lem:fullMainOneOutRegularDeletionAlgorithm} as a consequence of the following lemma.

\begin{lemma}\label{lem:mainOneOutRegularDeletionAlgorithm}
    Let $T\supseteq W\setminus S_1,\L\subseteq V(D)\setminus T\cup \{S_1\}$. 
	There is an algorithm that, given the tuple $(\tau,T,\L)$, runs in time $2^{O(k\log k)} \cdot n^{O(1)}$ and either returns a solution for $(D,k)$ or correctly concludes that there is no nice solution for $\tau$ that is also an $\L$-$T\cup \{S_1\}$ separator  and an $S_1$-$T$ separator.
\end{lemma}

\begin{proof}
We begin by checking whether $k<0$ or $\lambda(\L,T\cup \{S_1\})>k$ or $\lambda(S_1,T)>k$ or there is a strongly connected component of $D[W]$ that contains a vertex of out-degree at least 2. If any of these checks return an affirmative answer, then we return NO and terminate. Moreover, if $k=0$ and $(D,0)$ is a no-instance of {\sc 1-Out-Regular SCC Deletion} then we return NO and terminate. If $k\geq 0$ and $(D,0)$ is a yes-instance of {\sc 1-Out-Regular SCC Deletion} then we return $\emptyset$ and terminate.
% Similarly, we check whether $k\geq 0$, $\lambda(\L,T\cup \{S_1\})=0$, $\lambda(S_1,T)=0$ and $(D,0)$ is a yes-instance of {\sc 1-Out-Regular SCC Deletion}. If this check returns an affirmative answer, then we remove $S_1$ and recurse (after updating the indices of the remaining sets in $\S$) on the remaining instance.
Henceforth, we assume that $k\geq 1$, $\lambda(\L,T\cup \{S_1\})\leq k$, $\lambda(S_1,T)\leq k$ and every strongly connected component of $D[W]$ has maximum out-degree at most 1.  To provide an intuitive description of our algorithm, we fix a hypothetical nice solution $X^\star$ for $\tau$ that is also an $\L$-$T\cup \{S_1\}$ separator  and an $S_1$-$T$ separator (if one exists).

%For each of our branching rules, we give an informal description, which is followed by the formal description and analysis.

\begin{description}[leftmargin=*]
\item[Case I: $\L\neq \emptyset$.]~ 
If $\L$ is non-empty and  $\lambda(\L,T\cup \{S_1\})=0$, then we set $\L=\emptyset$ and recurse. 

Otherwise, we compute  $C_{max}=C_{max}(\L,T\cup \{S_1\})$. 
%If $|C|>k$, then we return NO. 
We now branch by  picking a vertex in $C_{max}$ and either guessing it to be in $X^\star$ or guessing that it is reachable from $\L$  in $D-X^\star$. The correctness of this branching follows from the fact that there is always a nice solution of the required type that is disjoint from $R_{max}(\L,T\cup \{S_1\})$ (Lemma~\ref{lem:disjointnessLemma}~(2)).

Formally, we do the following. We pick a vertex $c\in C_{max}$ and recurse on the tuple $(\tau_c,T,\L)$ where, $\tau_c=(D-\{c\},\S,W,k-1)$.
If this call returns a set $Z$, then we return the set $Z\cup \{c\}$.
% as a nice solution for $\tau$. 
 Otherwise, we recurse on the tuple $(\tau,T,\L\cup \{c\})$ and return its output. 

\noindent
		{\em Analysis.} Observe that in the first call, the budget $k$ drops by 1 (with $\lambda(\L,T\cup \{S_1\})$ and $\lambda(S_1,T)$ dropping by at most 1) while in the second call, $\lambda(\L,T\cup \{S_1\})$ increases by at least 1 (Lemma~\ref{lem:closest_sep_lambda_inc}).

 \smallskip
\item[Case II: $\L=\emptyset$.]~ 
We now describe the algorithm when $\L$ is empty.
\smallskip 
\begin{description}
\item[Case II.(a): $T=\emptyset$.]~  This also implies that $q=1$, i.e., $S_1=W$. Therefore, we may assume that every strongly connected component of $D$ intersects $S_1$ and moreover, contains at least one vertex with at least 2 out-neighbours in the same strongly connected component. 
%  and so, we may assume that $D$ \red{is strongly connected}. 
If the instance is not already solved and is a yes-instance, then there must exist
vertices $u_1,u_2,u_3$ such that  $u_2,u_3\in N^+(u_1)$ and $\{u_1,u_2,u_3\}\setminus S_1\neq \emptyset$. If such a triple of vertices do not exist, then we return NO as a nice solution for this tuple cannot exist (see Definition~\ref{def:outRegularniceInstances}).
 Otherwise, we branch by guessing one of $\{u_1,u_2,u_3\}$ to be contained in $X^\star$ ({subject to disjointness of the vertex from $S_1$}) or by guessing one of  $\{u_1,u_2,u_3\}$ to be unreachable from $S_1$ in $D-X^\star$ or by guessing that they are all reachable from $S_1$ in $D-X^\star$, in which case we guess that at least one of $\{u_1,u_2,u_3\}$ cannot reach $T\cup \{S_1\}$  in $D-X^\star$ (also {subject to disjointness from $S_1$}). The correctness of the branching follows from  Definition~\ref{def:outRegularniceInstances} and the fact that $X^\star$ is a nice solution for $\tau$.

 \smallskip
 
 Formally, we recurse on the tuples:
$\CC_1=(\tau_1,\emptyset,\emptyset)$,~$\CC_2=(\tau_2,\emptyset,\emptyset)$,~$\CC_3=(\tau_3,\emptyset,\emptyset)$,~$\CC_4=(\tau_4,\{u_1\},\emptyset)$,
~$\CC_5=(\tau_5,\{u_2\},\emptyset)$,
~$\CC_6=(\tau_6,\{u_3\},\emptyset)$,
~$\CC_7=(\tau,\emptyset,\{u_1\})$,~$\CC_8=(\tau,\emptyset, \{u_2\})$,~$\CC_9=(\tau,\emptyset, \{u_3\})$ defined as follows.
 For each $i\in [3]$, $\tau_i=(D-\{u_i\},\S,W,k-1)$ if $u_i\notin S_1$ and $\tau_i$ is a trivial no-instance otherwise. Notice that in the correct branch (say, Branch $i$), $X^\star\setminus \{u_i\}$ is a nice solution for $\tau_{i}$. 
  
 Similarly, $\tau_{3+i}=(D,(S_1,\{u_i\}),S\cup \{u_i\},k)$ if $u_i\notin S_1$ and $\tau_{3+i}$ is a trivial no-instance otherwise. Notice that in the correct branch (say, Branch $3+i$), $X^\star$ is a nice solution for $\tau_{3+i}$ that is an $S_1$-$u_i$ separator. 
  Finally, for each $i\in [3]$, $\tau_{6+i}=\tau$ if $u_i\notin S$ and $\tau_{6+i}$ is a trivial no-instance otherwise. 
 Let the respective outputs be denoted by $Z_1,\dots, Z_9$.
	 We return NO if all calls return NO. If any of the calls $\{\CC_i\mid i\in [3]\}$ returns a set $Z_i$, then we return $Z_i\cup \{u_i\}$ ($i$ is chosen to be the least such value) and terminate. If any of the calls $\{\CC_{3+i}\mid i\in [6]\}$ return a set $Z_i$, then we return the same set (with $i$ chosen to be the least such value) and terminate. 
	 
	 \noindent
{\em Analysis.}	Observe that in the first three calls, the budget $k$ drops by 1 (with no change in $\lambda(\L,T\cup \{S_1\})$ and $\lambda(S_1,T)$) while in the next three calls, $\lambda(S_1,T)$ increases by at least 1 and in the final three calls,	$\lambda(\L,T\cup \{S_1\})$ increases by at least 1.

\smallskip
\item[Case II.(b): $T\neq  \emptyset$.]~
If $\lambda(S_1,T)=0$ (recall that $\L=\emptyset$ and so, $\lambda(\L,T\cup \{S_1\})=0$), then we solve the problem independently on the subgraph induced by the strongly connected components intersecting $S_1$ and that induced by the remaining vertices.

Otherwise, we begin by running the polynomial-time algorithm of Lemma~\ref{lem:branchableObjectOutDegreeDeletion} to either  (i) correctly conclude that every vertex in $R_{max}=R_{max}(S_1,T)$ has out-degree at most 1 in $D$ or (ii) compute a $u_1\in R_{min}=R_{min}(S_1,T)$ and  $u_2,u_3\in N^+_D(u)$  
, or (iii) compute a minimum $S_1$-$T$ separator $C$ such that $R=R(S_1,C)$ has no vertex with out-degree at least 2 in $D$, and vertices $u_1\in C$, $u_2,u_3\in N^+_D(u_1)\setminus T$.

\smallskip
\begin{description}[leftmargin=*]
		 \item[Case II.(b).(i):]  We first guess a vertex of $close(S_1,R_{max})$ to be contained in $X^\star$. That is, for every $c\in close(S_1,R_{max})$, we recurse on the tuple $\widehat{\CC}_c=(\tau_c,T,\emptyset)$, where $\tau_c=(D-\{c\},\S,W,k-1)$. We will later prove that this branching factor is $O(k)$ because the size of  $S_1$, and hence, that of $close(S_1,R_{max})$  will always be bounded by $O(k)$. Let $\{Z_c\mid c\in close(S_1,R_{max})\}$ denote the outputs of these branches. If for any $c\in close(S_1,R_{max})$, $Z_c$ is a set, then we return $Z_c\cup \{c\}$ and terminate. Notice that in the correct branch (say, Branch-$c$), $X^\star\setminus \{c\}$ is a nice solution for $\tau_{c}$ that is also an $S_1$-$T$ separator in $D-\{c\}$.
		 
		 Finally, in the branch where $close(S_1,R_{max})$ is guessed to be disjoint from the solution, 
		 we pick a vertex $c\in C_{max}(S_1,T)$ and branch by either guessing it to be in $X^\star$ or by guessing $c$ to be reachable from $S_1$ in $D-X^\star$. The exhaustiveness follows from Lemma~\ref{lem:disjointnessLemma}~(1). 
		  In the latter case, we have two further branches where we either guess that $c$ (which, in this guess, cannot reach $T$ in $D-X^\star$) cannot reach $S_1$ or that $c$ can reach $S_1$ in $D-X^\star$. The latter two branches are executed by setting $\L=\{c\}$ or by adding $c$ to $S_1$ respectively. The correctness of adding $c$ to $S_1$ follows from Lemma~\ref{lem:disjointnessLemma}~(3).

		 \smallskip
Formally, we recurse on the tuples $\CC_1=(\tau_1,T,\emptyset)$,  $\CC_2=(\tau,T,\{c\})$  and $\CC_3=(\tau_2,T,\emptyset)$, 
	 	  where $\tau_1=(D-\{c\},\S,W,k-1)$ and $\tau_2=(D,(S_1\cup \{c\},S_2,\dots, S_q),W\cup \{c\},k)$. 
	This is now an {exhaustive branching} due to Lemma~\ref{lem:disjointnessLemma}~(1). Let $Z_1,Z_2,Z_3$ denote the respective outputs of the three final recursions. If $Z_2$ or $Z_3$ is a set, then we return this set (chosen arbitrarily) and terminate. On the other hand, if $Z_1$ is a set, then we return $Z_1\cup \{c\}$ and terminate. Finally, if all recursive outputs are NO, then we return NO and terminate. 
	
	\noindent
	{\em Analysis.} 
	We remark that this case contains the only branch where vertices are added to $S_1$. Recall that $S_1$ is initially bounded by $k$. We will prove later that the depth of the search tree is bounded by $O(k)$, implying that this branch can only be taken $O(k)$ times, and $|close(S_1,R_{max})|=O(k)$ as required (regardless of the value of $S_1$ at any given stage of the algorithm), thus bounding the branching factor. Moreover, in the first $|close(S_1,R_{max})|+1$ branches, the budget $k$ drops by 1 while $\lambda(S_1,T)$ decreases by at most 1.
		In the penultimate branch, $\lambda(\L,T\cup \{S_1\})$ strictly increases and in the final branch, $\lambda(S_1,T)$ strictly increases (Lemma~\ref{lem:closest_sep_lambda_inc}).
		 
		 	 \smallskip 
\item[Case II.(b).(ii):] 
In this case, we branch by either guessing one of $\{u_1,u_2,u_3\}$ to be in $X^\star$ ({subject to disjointness from $S_1$}) or by guessing that $u_1$ is unreachable from $S_1$ in $D-X^\star$ (in which case we add $u_1$ to $T$) or by guessing that $\{u_1,u_2,u_3\}$ are all reachable from $S_1$ in $D-X^\star$ but at least one of $\{u_1,u_2,u_3\}$ cannot reach $S_1\cup T$  in $D-X^\star$ (again, {subject to disjointness from $S_1$}). This branching is exhaustive due to Definition~\ref{def:outRegularniceInstances} and the fact that $X^\star$ is a nice solution for $\tau$. 

Formally, we recurse on the following tuples:  $\CC_1=(\tau_1,T,\emptyset)$, $\CC_2=(\tau_2,T,\emptyset)$, $\CC_3=(\tau_3,T,\emptyset)$, $\CC_4=(\tau_4,T\cup \{u_1\},\emptyset)$, $\CC_5=(\tau_5,T, \{u_1\}),\CC_6=(\tau,T,\{u_2\}),\CC_7=(\tau,T,\{u_3\})$ defined as follows.  For each $i\in \{1,2,3\}$, $\tau_i=(D-\{u_i\},\S,W,k-1)$ if $u_i\notin S_1$ and $\tau_i$ is a trivial no-instance otherwise. $\tau_4=\tau$ if $u_1\notin S_1$ and $\tau_4$ is a trivial no-instance otherwise. Finally, for each $i\in \{1,2,3\}$, $\tau_{4+i}=\tau$ if $u_i\notin S_1$ and $\tau_{4+i}$ is a trivial no-instance otherwise.

	 Let the outputs of these recursive calls be denoted by $Z_1,\dots, Z_7$ respectively.
	 We return NO if all calls return NO. If any of the calls $\{\CC_i\mid i\in [3]\}$ returns a set $Z_i$, then we return $Z_i\cup \{u_i\}$ ($i$ is chosen to be the least such value) and terminate.
	  If any of the calls $\{\CC_i\mid i\in \{4,5,6,7\}\}$ returns a set $Z_i$, then we return the same set (with $i$ chosen to be least possible) and terminate.

\noindent
{\em Analysis.}	 In the first three branches, the budget $k$ decreases by 1 while $\lambda_D(S_1,T)$ decreases by at most 1. In $\CC_4$, $\lambda(S_1,T)$ increases by at least 1 (Lemma~\ref{lem:closest_sep_lambda_inc}) and in the final three branches, $\lambda_D(\L,T\cup S_1)$ increases by at least 1. We may assume that every vertex in 
 $D$ can reach $T\cup S_1$. Otherwise, they can be deleted as they do not participate in $S_1$-$T$ paths or in strongly connected components containing a vertex of out-degree at least 2 (within the same strongly connected component).
%\red{What if there is no path from $u_i$ to $T\cup S_1$?}
%Recall that according to our assumption (\red{see discussion following Observation~\ref{obs:uniqueReachableVertexConsequence}}), every vertex in $D$ has a path to $W$. 

	 	 \smallskip 
\item[Case II.(b).(iii):] 
In this case, we execute a combination of the branchings used in Case II.(b).(i) and Case II.(b).(ii). That is, we first guess a vertex of $close(S_1,R)$ to be contained in $X^\star$. That is, for every $c\in close(S_1,R)$, we recurse on the tuple $\widehat{\CC}_c=(\tau_c,T,\emptyset)$, where $\tau_c=(D-\{c\},\S,W,k-1)$. In the remaining branch, we may assume that $close(S_1,R)$ is disjoint from $X^\star$ and hence, by Lemma~\ref{lem:disjointnessLemma}~(1), $X^\star$ is also disjoint from $R$.

We now branch by either guessing one of $\{u_1,u_2,u_3\}$ to be in $X^\star$ or by guessing that $\{u_1,u_2,u_3\}$ are all reachable from $S_1$ in $D-X^\star$ but at least one of $\{u_1,u_2,u_3\}$ cannot reach $S_1\cup T$  in $D-X^\star$. The correctness follows from Definition~\ref{def:outRegularniceInstances} and Lemma~\ref{lem:disjointnessLemma}~(1) and the analysis is identical to that used in the previous two cases.

%\red{We recurse on the tuples $\{\CC_i\mid i\in [7]\setminus \{4\}\}$ described in Case II.(b).(ii).}
%This branching is correct due to Definition~\ref{def:outRegularniceInstances} and Lemma~\ref{lem:disjointnessLemma}~(1). 

		 \end{description}

\end{description}

\end{description}

This completes the description of the algorithm. 

\smallskip
\noindent
\emph{Correctness and running time analysis.}~
The correctness follows from the terminating conditions of the algorithm outlined in the beginning, plus the fact that (a) every branching step has been argued to be exhaustive  and (b) the measure $3k-\lambda(S_1,T)-\lambda(\L,T\cup S_1)$ strictly decreases in every recursive call.   
Since the algorithm has $O(k)$ branches at any step, we conclude that the running time of this algorithm is  $2^{O(k\log k)} \cdot n^{O(1)}$. This completes the proof of the lemma.
\end{proof}

We obtain Lemma~\ref{lem:fullMainOneOutRegularDeletionAlgorithm} by invoking Lemma~\ref{lem:mainOneOutRegularDeletionAlgorithm} on the tuple ($(D,\S,W,k)$,$T=W\setminus S_1$,$\L=\emptyset$). This completes our proof of Theorem~\ref{thm:oneOutRegularDeletion}.

%\newpage

\subsection{Faster FPT algorithm for {\sc Bounded Size SCC Deletion}}

In this section, we give an algorithm for {\sc Rooted $\cH$-SCC Deletion} that runs in time $2^{O(k\log s} \cdot n^{O(1)}$ in the special case where $\cH$ is the set of all arborescences on exactly $s+1$ vertices, that is we prove Theorem~\ref{thm:boundedSizeDeletion}.
%  Notice that the strongly connected graphs that exclude the graphs in $\cH$ as subgraphs are precisely the strongly connected graphs of size at most $s$. Our result therefore improves upon that of G{\"{o}}ke et al.~\cite{GokeMM19} who gave an algorithm with running time $O^*(4^k(ks+k+s)!)$. Notice that when $s=1$, this is precisely the DFVS problem.

%. and moreover, for $s=1$, matches the running time of DFVS.  

\boundedSizeDeletion*

In the following, fix $\tau=(D,\S=(S_1,\dots$, $S_q),W,k)$ to be an instance of {\sc Rooted $\cH$-SCC Deletion}. Here, $\cH$ is the set of all arborescences on $s+1$ vertices.
%\begin{definition}[\blue{Nice Instances for Bounded Size SCC Deletion}]\label{def:BoundedSizeSCCDeletionniceInstances}
%%	Let $(D,\S=(S_1,\dots, S_q),W,k)$ be an instance of {\sc 1-out-regular deletion}. 
%	A solution $X$ for the instance $\tau$ is said to be {\em nice} if for every  triple $u,v,w\in V(D)$  such that $v,w\in N^+(u)$ and for every $i\in [q]$, one of the following holds.
%	
%	\begin{enumerate}
%	\item $X$ intersects $\{u,v,w\}$.
%	\item $u\notin R(S_i,X)$.
%	\item  There is no $v$-$S_i$ path in $D-X$ or no $w$-$S_i$ path in $D-X$.
%	\end{enumerate}
%
%\end{definition}
Using Lemma~\ref{lem:DisjCompressionToDisjPartionedCompression}, Theorem~\ref{thm:boundedSizeDeletion} can be obtained as a consequence of the following lemma.

\begin{lemma}\label{lem:fullMainBoundedSizeDeletionAlgorithm}
There is an algorithm that, given $\tau$, runs in time $2^{O(k\log s)} \cdot n^{O(1)}$ and either returns a solution for $(D,k)$ or correctly concludes that there is no nice solution for $\tau$.
\end{lemma}

%Therefore, it suffices to prove Lemma~\ref{lem:fullMainBoundedSizeDeletionAlgorithm}.

%
%
%\begin{lemma}\label{lem:branchableObjectBoundedSizeDeletion}Let $T\supseteq W\setminus S_1$,
%There is an algorithm that, given $\tau$ and $T$,  runs in polynomial time and performs one of the following operations: 
%\begin{enumerate}\item Outputs an arborescence of size $s+1$ contained in $D[R_{min}(S_1,T)]$. 
%\item 	Correctly concludes that there is some \red{furthest} $S_1$-$T$ separator $C$ of size at most $k$ such that $|R(S_1,T)|\leq s$.
%\item Outputs a pair of $S_1$-$T$ separators $C_1,C_2$ of size at most $k$ such that $C_2$ tightly covers $C_1$, $|R(S_1,C_1)|\leq s<|R(S_1,C_2)|$. That is, 
%$R(S_1,C_1)$ has  no arborescence of size $s+1$ (i.e., $|R(S_1,C_1)|\leq s$) and $R(S_1,C_2)$ contains an arborescence of size $s+1$.
%\end{enumerate}
%\end{lemma}
%
%The algorithm of the above lemma is a straightforward greedy algorithm that begins by computing $C_{min}(S_1,T)$ and then greedily proceeding to find successive (under the tight-covering relation) separators of size at most $k$, until one of the three exhaustive conditions in the statement of the lemma is identified to hold.
%

The proof of the following lemma is identical to that of Lemma~\ref{lem:disjointnessLemma}~(2).
%\todo{lemma 38(2)}.

\begin{lemma}\label{lem:disjointnessLemmaBoundedSizeDeletion}
%Let $\tau=(D,\S=(S_1,\dots, S_q),W,k)$ be an instance of {\sc 1-Out-Regular SCC Deletion}. 
%Let $T\supseteq W\setminus S_1$,
Let $\L\subseteq V(D)\setminus W$.
	Suppose that there is a nice solution $X$ for $\tau$ that is an $\L$-$W$ separator. 
%	Then the following statements hold:
%	\begin{enumerate}	\item 	
	Then, there is a nice solution $X'$ for $\tau$ that is an $\L$-$W$ separator and moreover, $X'\cap R_{max}(\L,W)=\emptyset$. 
%	\item If $c\in V(D)$ is  reachable from $S_1$ and can reach $S_1$ in $D-X$, then $X$ is also a nice solution for the instance $\tau'=(D,\S=(S_1\cup \{c\},\dots$, $S_q),W,k)$ that is an $S_1\cup \{c\}$-$T$ separator and an $\L$-$T\cup \{S_1\}\cup \{c\}$ separator. 
%	\end{enumerate}
%\red{Address the possibility that $\L\cap C\neq \emptyset$.}
\end{lemma}

\begin{lemma}\label{lem:mainBoundedSizeDeletionAlgorithm}
%    Let $T\supseteq W\setminus S_1,
Let     $\L\subseteq V(D)\setminus W$. 
	There is an algorithm that, given the tuple $(\tau,\L)$, runs in time $2^{O(k\log s)} \cdot n^{O(1)}$ and either returns  a solution for $(D,k)$ or correctly concludes that there is no nice solution for $\tau$ that is also an $\L$-$W$ separator.
	\end{lemma}

\begin{proof}
This algorithm closely resembles parts of the algorithm of Lemma~\ref{lem:mainOneOutRegularDeletionAlgorithm}. Therefore, we only give a high level sketch outlining the differences. Fix a nice solution $X^\star$ that is also an $\L$-$W$ separator (if one exists). 

	The base cases  and Case I ($\L\neq \emptyset$) are identical to those in Lemma~\ref{lem:mainOneOutRegularDeletionAlgorithm}, with the correctness of the latter following from Lemma~\ref{lem:disjointnessLemmaBoundedSizeDeletion}. In Case II ($\L=\emptyset$), we pick an arborescence on $s+1$ vertices that is rooted at a vertex $c\in S_1$ and  and branch by guessing one of these vertices to be in $X^\star$, in which case we delete it and recurse, or by guessing that one of these vertices cannot reach $W$ in $D-X^\star$, in which case we add it to $\L$ and recurse. This is correct because any arborescence on $s+1$ vertices that is reachable from a vertex of $S_1$ implies the existence of an arborescence on $s+1$ vertices that is rooted at a vertex of $S_1$. 
	Notice that the branching factor in either case is bounded by $2(s+1)$ and in each branch, the measure $2k-\lambda(\L,W)$ strictly decreases. To achieve the increase in $\lambda(\L,W)$ in the second set of $s+1$ branches of Case II (where $\L$ changes from $\emptyset$ to a singleton), we note that any vertex of $D$ that cannot reach $W$ in $D$ cannot be part of either a cycle or an $S_i$-$S_j$ path for some $i,j$ and hence, can be removed without affecting the presence of a nice solution. 
	 This completes the proof of the lemma.
\end{proof}

We obtain Lemma~\ref{lem:fullMainBoundedSizeDeletionAlgorithm} by invoking Lemma~\ref{lem:mainBoundedSizeDeletionAlgorithm} on the tuple ($(D,\S,W,k)$,$\L=\emptyset$). This completes our proof of Theorem~\ref{thm:boundedSizeDeletion}.

\section{Conclusions}
We have initiated the study of the parameterized complexity of {\sc $\cH$-SCC Deletion} problem, where the objective is to compute a maximum subdigraph where no strongly connected component contains a  forbidden subgraph from the family $\cH$. This problem is a natural generalization of the classic {\sc Directed Feedback Vertex Set} problem.
We have obtained fixed-parameter algorithms for this problem when $\cH$ either contains at least one path or only contains rooted graphs.  
Furthermore, we have demonstrated that for a pair of previously studied special cases of this problem, one can obtain faster fixed-parameter algorithms by using our general strategy tailored to these special cases.

Our algorithms are fixed-parameter tractable parameterized by $k$ for {fixed} families $\cH$ (that contain a path or only rooted digraphs) and a fixed-parameter algorithm for this problem parameterized by both $k$ and $d_\cH$ (size of the largest graph in $\cH$) is unlikely to exist in general. 
%This is because the {\sc $\cH$-SCC Deletion} problem can be easily seen to be W[1]-hard parameterized by $d$ for $k=0$.

Our work identifies some natural directions for future research. In particular, can we completely characterize those finite families $\cH$ for which this problem is fixed-parameter tractable? Furthermore, one could ask: for which  infinite families $\cH$ does this problem admit a fixed-parameter algorithm? Recently G\"{o}ke, Marx and Mnich~\cite{gokeicalp2020} gave a fixed-parameter algorithm for the case where  $\cH$ is the set of all cycles of length {\em at least} a given $s$. 

\paragraph*{Acknowledgements.} The second author thanks Matthias Mnich for insightful discussions during  Dagstuhl Seminar 19041 and for the pointer to \cite{GokeMM19}.

\bibliographystyle{plain}
\bibliography{HSCC}

\end{document}